\documentclass[10pt, final, journal, letterpaper, twocolumn]{IEEEtran}
\makeatletter
\def\ps@headings{%
	\def\@oddhead{\mbox{}\scriptsize\rightmark \hfil \thepage}%
	\def\@evenhead{\scriptsize\thepage \hfil \leftmark\mbox{}}%
	\def\@oddfoot{}%
	\def\@evenfoot{}}
\makeatother 

\IEEEoverridecommandlockouts
\usepackage{bbm}
\usepackage{amsfonts}
\usepackage[dvips]{graphicx}
\usepackage{times}
\usepackage{cite}
\usepackage{amsmath}
\usepackage{array}
\usepackage{amssymb}

\usepackage{stfloats}
\usepackage{slashbox}
\usepackage{graphicx}
\usepackage{footnote}
\usepackage{booktabs}
\usepackage{array}
\usepackage{algorithm}
\usepackage{subeqnarray}
\usepackage{cases}
\usepackage{threeparttable}
\usepackage{color}
\usepackage{hyperref}
\usepackage{epstopdf}
\usepackage{algpseudocode}
\usepackage{bm}
\usepackage{multirow}
\usepackage[labelformat=simple]{subcaption}
\usepackage{adjustbox}

\newtheorem{theorem}{Theorem}

\newtheorem{corollary}{Corollary}

\begin{document}
	
	\title{IRS Aided Millimeter-Wave Sensing and Communication: Beam Scanning, Beam Splitting, and Performance Analysis}

	\author{\authorblockN{Renwang Li, {Xiaodan Shao,~\IEEEmembership{Member,~IEEE},  Shu Sun, \IEEEmembership{Member,~IEEE}, Meixia Tao, \IEEEmembership{Fellow,~IEEE}, Rui Zhang, \IEEEmembership{Fellow,~IEEE}}\\
			\thanks{Part of this work was presented at the 2023 IEEE 24th International Workshop on Signal Processing Advances in Wireless Communications (SPAWC) \cite{renwang}.}
			\thanks{R. Li, S. Sun, and M. Tao are with  the Department of Electronic Engineering and the 		Cooperative Medianet Innovation Center (CMIC), Shanghai Jiao Tong University, Shanghai 200240, China (e-mails:\{renwanglee, shusun, mxtao\}@sjtu.edu.cn).}
			\thanks{X. Shao is with the Institute for Digital Communications, Friedrich-Alexander-University Erlangen-Nuremberg, 91054
				Erlangen, Germany (e-mail: xiaodan.shao@fau.de).} 
			\thanks{ R. Zhang is with School of Science and Engineering, Shenzhen Research Institute of Big Data, The Chinese University of Hong Kong, Shenzhen, Guangdong 518172, China (e-mail: rzhang@cuhk.edu.cn). He is also with the Department of Electrical and Computer Engineering, National University of Singapore, Singapore 117583 (e-mail: elezhang@nus.edu.sg).
			}
		}
	}
	
	\maketitle
	
	\begin{abstract}
		Integrated sensing and communication (ISAC)   has attracted growing 
		interests for enabling the future 6G wireless networks, due to its capability of sharing spectrum and hardware resources between communication and sensing systems. However, existing works on ISAC usually need to modify the communication protocol to cater for the new sensing performance requirement, which may be difficult to implement in practice.   In this paper, we study a new  intelligent reflecting surface (IRS) aided millimeter-wave (mmWave) ISAC system by exploiting the distinct beam scanning operation in mmWave communications to achieve efficient sensing at the same time.  First, we propose a two-phase ISAC  protocol aided by a semi-passive IRS, consisting of beam scanning and data transmission. Specifically, in the beam scanning phase, the IRS finds the optimal beam for reflecting signals from the base station to a communication user via its passive elements. Meanwhile, the IRS directly estimates the angle of a nearby target based on echo signals from the target using its equipped active sensing element. Then, in the data transmission phase, the sensing accuracy is further improved by leveraging the data signals via possible  IRS beam splitting.  Next, we derive the achievable rate of the communication user as well as the Cram\'er-Rao bound and the approximate mean square error of the target angle estimation  Finally, extensive simulation results are provided to verify our analysis as well as the effectiveness of the proposed scheme.
	\end{abstract}
	
	\section{Introduction}
	Recently, integrated sensing and communication (ISAC) has been recognized as a key technology for the future 6G wireless network due to its potential to enable efficient sharing of spectrum and hardware resources between communication and sensing systems \cite{ 9705498,9737357}. Meanwhile,   millimeter-wave (mmWave) technology can provide high data rate for communication  as well as high resolution for sensing, making it promising for realizing ISAC systems. However, mmWave signals are susceptible to   obstacles, and the performance of  mmWave ISAC systems can degrade dramatically in the absence of line-of-sight (LoS) path. To overcome this issue, intelligent reflecting surface (IRS)  has been recognized as a practically viable solution \cite{8910627, 9326394,IRS2}.  IRS is generally a digitally-controlled metasurface composed of a large number {of} passive reflecting elements (REs) that can reflect the incident signal with independently controlled phase shifts. By leveraging  IRS, a virtual LoS link can be created between two wireless nodes when their direct link is obstructed, thus allowing for uninterrupted  sensing and communication.
	
	Motivated by the above, significant research  efforts have been devoted to studying IRS-aided ISAC systems  \cite{9364358, 9771801, 9416177,song2022intelligent, 9769997, 9729741, 10038557,9593143}.  In \cite{9364358}, the joint design of transmit beamforming at the base station (BS) and reflection coefficients at the IRS is studied to maximize the signal-to-noise ratio (SNR) of radar detection while meeting communication requirements simultaneously.  The works \cite{9771801} and \cite{9416177} address radar beampattern design problems in single-user and multi-user scenarios, respectively.   The Cram\'er-Rao bound (CRB) minimization for IRS-aided sensing is considered in \cite{song2022intelligent}. The authors in \cite{9769997} aim to maximize the radar output signal-to-interference-plus-noise ratio (SINR) while guaranteeing the communication quality. A double-IRS-aided communication radar coexistence system is considered in \cite{9729741}. In \cite{10038557}, a feedback-based beam training approach is proposed to design BS transmit beamforming and IRS reflection coefficients for simultaneous communication and sensing. The authors in \cite{9593143} propose a multi-stage hierarchical beam training codebook to achieve the desired  accuracy for IRS-aided localization while ensuring a reliable communication link with the user. Notice that all of the aforementioned works adopt passive IRS to assist sensing, and thus their performance  is hindered by the severe path loss of the BS-IRS-target-IRS-BS cascaded echo link, particularly in mmWave frequencies. 
	
	{It is worth mentioning that  existing studies on IRS-aided ISAC \cite{9364358, 9771801, 9416177,song2022intelligent, 9769997, 9729741} have mostly assumed that the target angle from the IRS is  known a prior within a certain range and mainly focused on the data transmission period; however, there has been very limited consideration of exploiting the channel estimation/training period for achieving ISAC \cite{9838563}. Moreover, the aforementioned works usually require protocol modifications to accommodate the new sensing performance requirements, a task that might pose challenges in practical implementation.}  In practice, mmWave communication systems typically adopt the transmission protocol with two phases, namely,  the beam scanning/training phase and the data transmission phase \cite{8458146}. In the beam scanning phase, the BS transmits reference/pilot signals using different beams from a given codebook, while the user measures the received signal power for each beam and feeds back the index of the beam with the maximum power to the BS. Subsequently, the BS adopts this maximum-power beam to transmit data during the data transmission phase. The above beam scanning protocol can be extended to work for IRS-aided mmWave communication systems, by applying firstly BS beam scanning to find the maximum-power beam towards the IRS, and then IRS beam scanning to find that towards the user, for the typical scenario where the LoS channel between the BS and user is severely blocked.   However, the exploration of the above protocol for target sensing  as well as the performance tradeoff between sensing and communication under this protocol remain unaddressed yet.
	
	As such, in this paper, we investigate a downlink IRS-aided mmWave ISAC system, as illustrated in Fig. \ref{fig_system}, where a ``semi-passive"  IRS consisting of passive REs and active sensing elements (SEs)  is adopted to create virtual LoS channels for the IRS to forward information from the BS to a nearby communication user as well as detect the angle of a nearby target. In particular, the  SEs are used to collect the echo signals reflected from the target for its angle estimation. Compared with a fully passive IRS that reflects the echoes from the target back to the BS for detection, the semi-passive IRS can directly estimate the angle of the target and thus  significantly reduce the path loss of the received echo signal at the BS \cite{9724202,Semi-IRS1,Semi-IRS2}. The main contributions of this paper are summarized as follows:
	\begin{itemize}
		\item First, we propose a two-phase ISAC protocol for the considered IRS-aided mmWave ISAC system, based on the practical two-phase communication protocol for mmWave systems. Specifically, in the beam scanning phase, the training signals from the BS are  used not only to identify the best IRS beam for the communication user, but also to initially estimate the target's angle from the IRS. Then, in the subsequent data transmission phase, the data signals from the BS are also used to improve the angle estimation   accuracy via possible  IRS beam splitting, while ensuring the achievable rate of the communication user.
		\item Second, we analyze the achievable rate of the 	communication user and derive the CRB of angle estimation for the sensing target in the beam scanning phase. The CRB analysis reveals that more REs and  SEs can achieve more accurate target angle estimation.  Additionally, we derive an analytical approximation of the mean square error (MSE) of the angle estimation, which leads to a closed-form expression for the minimum SNR  required to achieve a desired initial angle estimation accuracy.
		\item Third, in order to the enhance sensing accuracy in the data transmission phase, we propose two IRS beam design and sensing strategies, i.e., single-beam-based sensing and beam-splitting-based sensing, which are applied when the difference between the estimated angles of the communication user and the target from the IRS in the beam scanning phase is smaller than a given threshold and otherwise, respectively.    In the latter case, the loss in the achievable rate of the communication user is also characterized to ensure its performance.
	\end{itemize}	
	\begin{figure}[t]
	\begin{centering}
		\includegraphics[width=.42\textwidth]{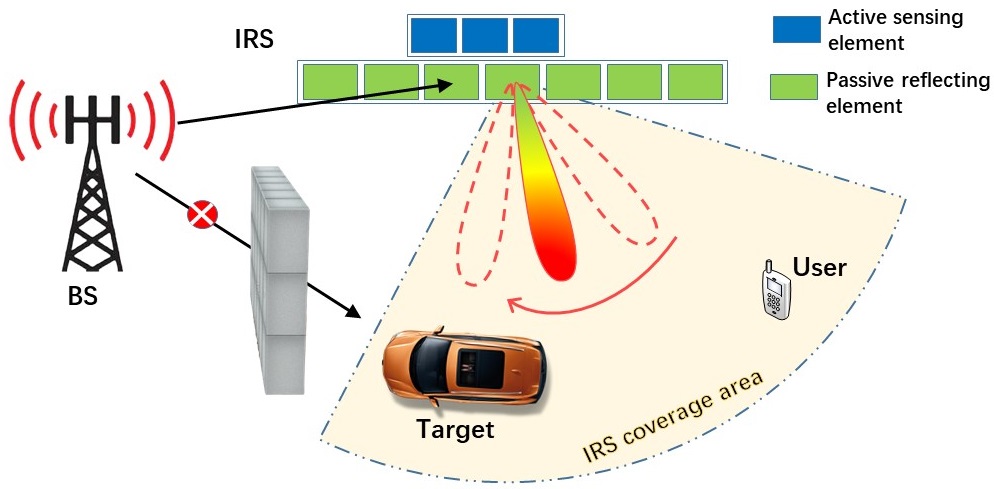}
		\caption{System model of IRS-aided ISAC.}\label{fig_system}
	\end{centering}
\vspace{-0.5cm}
\end{figure}	

	The rest of this paper is organized as follows. Section II introduces the IRS-aided mmWave system and presents the proposed ISAC protocol. Section III analyzes the communication and sensing performance in the beam scanning phase. Section IV proposes enhanced sensing strategies during the data transmission phase, with the resulted communication/sensing performance characterized. The simulation results are provided in Section V, and the paper is concluded in Section VI.
	
	\emph{Notations}: The imaginary unit is denoted by $j=\sqrt{-1}$. Vectors and matrices are denoted by bold-face lower-case and upper-case letters, respectively. $\mathbb{C}^{x\times y}$ denotes the space of $x\times y$ complex-valued matrices. $ \bf x^*$, $\mathbf{x}^T$, and $\mathbf x^H$ denote the conjugate, transpose and  conjugate transpose of vector $\bf x$. $\bf I$ denotes an identity matrix of appropriate dimensions.  $\operatorname{diag}(\mathbf{x})$ denotes a diagonal matrix with each diagonal element being the corresponding element in $\mathbf{x}$.   $\dot{\mathbf{a}}(\theta)$ denotes the gradient vector of ${\mathbf{a}}(\theta)$. $\operatorname{vec}(\cdot)$ denotes the vectorization operator. $\lfloor x \rfloor$ denotes the flooring operation that takes the largest integer no greater than $x$. The distribution of a circularly symmetric complex Gaussian (CSCG)  random vector with zero means and covariance matrix $\Sigma$ is denoted by $\mathcal{C}\mathcal{N}(\boldsymbol{0},\Sigma)$; and $\sim$ stands for ``distributed as''. The main notations used in this paper are summarized in Table \ref{table_notation}.

\begin{table*}[h!] 
	\begin{center}
		\caption{Summary of Notations.}
		\label{table_notation}
		\begin{tabular}{|c|p{4cm}||c|p{4cm}|}
			\hline
			\textbf{Notation} & \textbf{Description} &\textbf{Notation} & \textbf{Description} \\
			\hline
			$N$ & Antenna number of BS & $M$ & Number of IRS REs \\
			\hline
			$M_s$ & Number of IRS SEs & $M_e$ & Number of IRS REs allocated for target estimation in Phase II \\
			\hline
			$L$ & Codebook size & $\tau$ & Time duration of beam scanning \\
			\hline
			$T$ & Channel coherence time normalized to number of symbol durations & $\theta_{BI}$ & Spatial AoA from BS to IRS\\
			\hline 
			$\theta_{IU}$ & Spatial AoD from IRS to the target & $\theta_{IT}$ & Spatial AoD from IRS to the user \\
			\hline
			$\alpha_g$ & Path gain of BS-REs channel& $\alpha_h$ & Path gain of IRS-user channel \\
			\hline
			$\alpha_s$ & Path gain of REs-target-SEs link & $\ell$ & IRS's best beam index for the user\\
			\hline 
			$\mathbf{a}_r(\cdot)\in \mathbb{C}^{M\times 1}$ & Array response vector of REs  & $\mathbf{a}_s(\cdot)\in \mathbb{C}^{M_s\times 1}$ & Array response vector of SEs \\
			\hline
			$\mathbf{q}(\cdot)\in \mathbb{C}^{M\times 1}$ & $\mathbf{q}(\theta_{IT}) \triangleq \mathbf{a}_r (\overline{\theta}_{IT})$ & $\boldsymbol{\phi} \in \mathbb{C}^{M\times 1}$ & Reflection vector of the REs \\
			\hline
		\end{tabular}
	\end{center}
\vspace{-0.5cm}
\end{table*}

	\section{System Model and Protocol Design}		
	In this section, we introduce the IRS-aided ISAC system model and the corresponding channel model, and then propose the ISAC protocol.
	\subsection{System Model}
	We consider a downlink mmWave ISAC system with the aid of a semi-passive IRS as illustrated in Fig. \ref{fig_system}, where an $N$-antenna BS aims to communicate with a single-antenna user  and also   to detect the angle of a sensing target. The direct links between the BS and the user, as well as the target, are assumed to be blocked due to unfavorable propagation environment. Thus, the IRS is deployed to create virtual links for both communication and sensing. We consider the use of semi-passive IRS consisting of $M$ passive REs to reflect the transmitted signals from the BS to the user and target, and $M_s$ active  SEs to collect the echo signals from the target for its angle estimation. The complex-valued baseband transmitted signal at the BS can be expressed as $\mathbf{x}=\mathbf{w}s$, where  $s$ denotes the training/data symbol for the communication user with unit power and $\mathbf{w} \in \mathbb{C}^{N\times 1}$ is the transmit beamforming vector with $\|\mathbf{w}\|^2=1$. Then, the received signal $y_u$ at the user can be expressed as
	\begin{equation} \label{exp_user}
		y_u = \sqrt{P_t}\mathbf{h}_u^H \operatorname{diag} (\boldsymbol{\phi}) \mathbf{G} \mathbf{w} s + n_u, \\
	\end{equation}
	where $P_t$ is the transmit power at the BS, $\mathbf{G}\in \mathbb{C}^{M\times N}$ represents the channel between the BS and   REs, $\mathbf{h}_u\in \mathbb{C}^{M\times 1}$ represents the channel between the   REs and the user, $n_u \sim \mathcal{CN}(0, \sigma^2)$ is the receiver AWGN with $\sigma^2$ representing the noise power, and $\boldsymbol{\phi}\in \mathbb{C}^{M\times 1}$ represents the reflection vector at the   REs, which can be written as
	\begin{equation} 
		\boldsymbol{\phi}=\left[e^{j\phi_1},e^{j\phi_2},\ldots, e^{j\phi_M}  \right]^T,
	\end{equation}
	with $\phi_i$ being the phase shift by the $i$-th RE.
	
	The   SEs can simultaneously receive the signals transmitted  from  the BS and the echo signals reflected by the target\footnote{The radar cross section (RCS) of the communication user (terminal) is usually significantly smaller compared to the target. Hence, the echo signal reflected by the communication user can be safely ignored in the target angle estimation.}. In general, the angles of the target and BS with respect to the IRS are different and can be estimated by the   SEs based on the received echoes. The angle between the BS and IRS can be determined in advance by the   SEs, which facilitates the estimation of the target's angle. The received signal $\mathbf{y}_s \in \mathbb{C}^{M_s\times 1}$ at the   SEs can be represented as
	\begin{equation} \label{exp_sensors}
		\begin{aligned}
			\mathbf{y}_s &= \sqrt{P_t} \left(\mathbf{H}_t \operatorname{diag} (\boldsymbol{\phi}) \mathbf{G}+ \mathbf{G}_s\right) \mathbf{w} s   + \mathbf{n}_s, \\
		\end{aligned}
	\end{equation}
	where $\mathbf{H}_t \in \mathbb{C}^{M_s\times M}$ denotes the channel matrix of the   REs-target-SEs link, $\mathbf{G}_s \in \mathbb{C}^{M_s\times N} $ denotes the channel matrix of the BS-SEs link, and $\mathbf{n}_s \sim \mathcal{CN}(0, \sigma^2 \mathbf{I}_{M_s})$ is the receiver AWGN. 
	
	\subsection{Channel Model}
	We adopt the LoS channel model to characterize the 	mmWave channel. For ease of exposition, we assume that uniform linear arrays (ULAs)  are equipped at the BS, REs, and  SEs. Thus, the BS-REs channel can be expressed as
	\begin{equation} 
		\mathbf{G} = \alpha_g \mathbf{a}_r (\theta_{BI}) \mathbf{a}_b^H (\vartheta_{BI}),
	\end{equation}
	where $\alpha_g=\frac{\lambda}{4\pi d_{BI}} e^{\frac{j 2\pi d_{BI}}{\lambda}}$ \cite{9724202} denotes the complex-valued path gain of the BS-REs channel with $\lambda$ being the carrier wavelength and $d_{BI}$ being the distance between the BS and IRS,  $\vartheta_{BI}=\sin(\varsigma_{BI})$ with $\varsigma_{BI}$ denoting the angle of departure (AoD) from the BS, $\theta_{BI} = \sin(\zeta_{BI})$ with $\zeta_{BI}$ denoting the angle of arrival (AoA) to the IRS, and $\mathbf{a}_r(\cdot)$ $\left(\mathbf{a}_b(\cdot)\right)$ denotes the array response vector associated with the REs (BS). The array response vector for a ULA with $M$ elements of half-wavelength spacing  and the center of the ULA  as the reference point  can be expressed as
	\begin{equation} \label{exp_array_resp}
	{\mathbf{a}}(\theta) = \left[
		e^{-j \frac{(M-1)\pi \theta }{2}},  e^{-j \frac{(M-3)\pi \theta }{2}},   \ldots,  e^{j \frac{(M-1)\pi \theta}{2}}
	\right]^{T}.
\end{equation}
	
	The IRS-user channel $\mathbf{h}_u$ can also be written as
	\begin{equation} 
		\mathbf{h}_u = \alpha_h \mathbf{a}_r (\theta_{IU} ),
	\end{equation}
	where $\alpha_h=\frac{\lambda}{4\pi d_{IU}} e^{\frac{j 2\pi d_{IU}}{\lambda}}$ denotes the complex-valued path gain of the IRS-user channel with $d_{IU}$ being the distance between the IRS and user, and $\theta_{IU}=\sin (\zeta_{IU})$ with $\zeta_{IU}$ denoting the AoD associated with the IRS.	
	
	The BS-SEs channel $\mathbf{G}_s$ can be represented as
	\begin{equation} 
		\mathbf{G}_s = \alpha_g \mathbf{a}_s ( \theta_{BI}) \mathbf{a}_b^H (\vartheta_{BI}),
	\end{equation}	
	where $\mathbf{a}_s(\cdot)$  denotes the array response vector associated with the SEs. 
	The REs-target-SEs channel $\mathbf{H}_t$ can be expressed as
	\begin{equation} 
		\mathbf{H}_t = \alpha_s \mathbf{a}_s ( \theta_{IT}) \mathbf{a}_r^H (\theta_{IT} ),
	\end{equation}
	where  $\theta_{IT}= \sin(\zeta_{IT})$ with $\zeta_{IT}$ denoting AoD from the SEs,  $\alpha_s =\sqrt{\frac{\lambda^2 \kappa}{64\pi^3 d_{IT}^4}} e^{\frac{j 4\pi d_{IT}}{\lambda}}$ refers to the complex path gain of the  REs-target-SEs link \cite{9367457}, in which  $d_{IT}$ denotes the distance between the IRS and target and $\kappa$ denotes the RCS of the target.
	
	Considering that the locations of the BS and IRS are fixed, the BS-REs channel $\mathbf{G}$ and the BS-SEs channel $\mathbf{G}_s$ are assumed to be constant for a long period, and can be estimated beforehand at the BS to achieve the optimal transmit beamforming as $\mathbf{w}=\frac{1}{\sqrt{N}} \mathbf{a}_b (\vartheta_{BI})$. Thus, in this paper we focus on the  beam training at the IRS. As a result, the received signal at the communication user  in \eqref{exp_user}  can be rewritten as
	\begin{equation} \label{exp_user2}
		y_u = \sqrt{N P_t} \alpha_g \mathbf{h}_u^H \operatorname{diag} (\boldsymbol{\phi}) \mathbf{a}_r (\theta_{BI}) s + n_u \\
	\end{equation}
	and the received   signal at the   SEs in \eqref{exp_sensors} can be rewritten as
	\begin{equation} \label{exp_sensors2}
	\hspace{-0.1cm}	\mathbf{y}_s 
		= \sqrt{N P_t } \alpha_g \left( \mathbf{H}_t \operatorname{diag} (\boldsymbol{\phi}) \mathbf{a}_r (\theta_{BI})+ \mathbf{a}_s(\theta_{BI}) \right) s  + \mathbf{n}_s.
	\end{equation}
	Note that  there exists an undetectable region, denoted as $\{\Omega_u | \theta_{IT}: |\theta_{IT}-\theta_{BI}|< \frac{2}{M_s} \}$, associated with the target.  Specifically, the   SEs receive strong   signals from the BS and weak echo signals from the target. When the angles between the target and BS with respect to the IRS are close, the echo signals from the target cannot be effectively extracted from the mixed signals, and thus the target cannot be detected. This fact will be elaborated analytically in Section \ref{sec_mse_phaseI}. Therefore, our focus is on the scenarios where the target is located outside the undetectable region.

	\subsection{Proposed Protocol for ISAC}
	\begin{figure}[t]
		\begin{centering}
			\includegraphics[width=.45\textwidth]{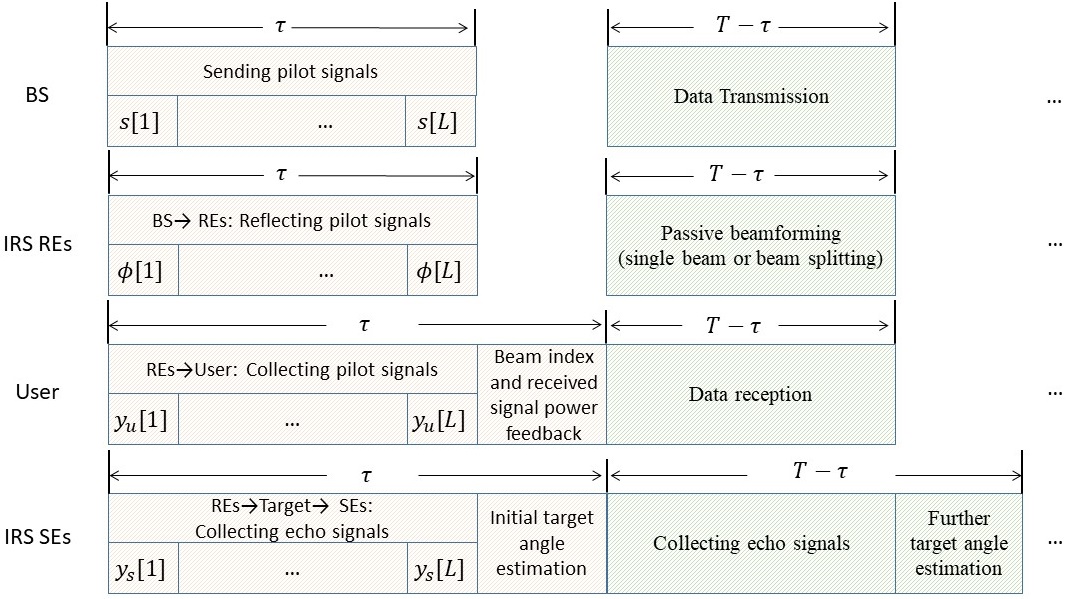}
			\caption{IRS-aided mmWave ISAC protocol.}\label{fig_protocol}
		\end{centering}
	\vspace{-0.5cm}
	\end{figure}
	
	In this subsection, we propose a two-phase protocol for the considered IRS-aided mmWave ISAC system. Following the existing mmWave communication protocol, beam training/scanning needs to be first conducted at the IRS, followed by data transmission. During the beam scanning phase, we adopt the widely-used discrete Fourier transformation (DFT) codebook $\mathbf{D}$ with $L$ beams as  follows,
	\begin{equation} \label{exp_codebook}
		\mathbf{D} \triangleq[\mathbf{a}_r (\eta_1), \mathbf{a}_r(\eta_2), \cdots, \mathbf{a}_r(\eta_L)] \in \mathbb{C}^{M \times L},
	\end{equation}
	where $\eta_i=-1+\frac{2 i-1}{L}, i=1, \cdots, L$, and $L \geq M$. In our system model, the SEs can exploit downlink beam scanning for sensing by collecting the echo signals reflected from the target. In addition, the communication signals during the data transmission phase can also be utilized to refine the target sensing performance (see Section \ref{sec_target_enhancement} for details).  {The ISAC protocol, depicted in Fig. \ref{fig_protocol},  is divided into two phases, with $T$ denoting the channel coherence time normalized to number of symbol durations. The first phase involves beam scanning with $\tau$ symbol durations, while the second phase focuses on data transmission with $T-\tau$ symbol durations. The time of beam index and received signal power feedback as well as target angle estimation is ignored. }
	
	\begin{itemize}
		\item Phase I (beam scanning): The BS sends downlink training signals. The REs sweep the beams in the codebook $\mathbf{D}$, while the SEs collect the echo signals reflected from the target. At the end of IRS beam scanning, the communication user identifies the IRS's best beam and corresponding received power, and feeds them back to the IRS controller (directly or via the BS). Meanwhile, the SEs estimate the target angle based on the received echo signals.
		\item Phase II (data transmission and enhanced sensing): The BS sends downlink data signals. If the target  is detected to be located in the vicinity of the communication user, REs reflect the data signal towards the user with the IRS's best beam found in Phase I, which is also reflected towards the target to enhance the estimation accuracy (thus termed as IRS single-beam). On the other hand, if the user and target are detected to be well separated, IRS beam splitting is adopted where a certain portion of REs are used for target  sensing  and the remaining REs are for communication with their corresponding optimal beam, provided that the achievable rate of the user is ensured within an acceptable margin. At the end of data transmission, the  SEs further estimate the target angle based on the received echo signals.
	\end{itemize}

	\section{Performance Analysis for Initial Beam Scanning}
	\subsection{Achievable Rate of Communication User} \label{sec_rate_phaseI}
	In this subsection, we first analyze the maximum channel gain obtained in the beam scanning phase, and then derive the achievable rate of the communication user.  {Assume the duration of one beam is equal to one symbol duration for simplicity, we have $\tau=L$. The BS's reference signal can be set as $s[t] =1, \forall t=1,2,\ldots,L$.} The received signal at the user in \eqref{exp_user2} can be expressed as
		\begin{align}
			y_u[t] &= \sqrt{N P_t} \alpha_g \mathbf{h}_u^H \operatorname{diag} (\boldsymbol{\phi}[t]) \mathbf{a}_r (\theta_{BI})  + n_u[t] \notag \\
			& = \sqrt{N P_t} \alpha_g \alpha_h \mathbf{a}_r^H (\theta_{IU}) \operatorname{diag} (\boldsymbol{\phi}[t]) \mathbf{a}_r (\theta_{BI})  + n_u[t] \notag\\
			& = \sqrt{N P_t} \alpha_g \alpha_h \boldsymbol{\phi}^T[t] \operatorname{diag} (\mathbf{a}_r^H (\theta_{IU}) ) \mathbf{a}_r (\theta_{BI})  +n_u[t] \notag\\
			& = \sqrt{N P_t} \alpha_g \alpha_h \left(\mathbf{a}_r^H (\overline{\theta}_{IU}) \boldsymbol{\phi}[t] \right) +n_u[t],
		\end{align}
	where $\overline{\theta}_{IU}= \theta_{IU} - \theta_{BI}$, and $\boldsymbol{\phi}[t] \in \mathbf{D}$. We assume that $\ell$ is the best beam index, i.e.,
	\begin{equation} 
		\ell = \arg \max \limits_{t, t=1, \cdots, L} |y_u[t]|^2.
	\end{equation}	
	Let $\delta_u= \left|\overline{\theta}_{IU} - \eta_\ell \right| $ denote  the spatial direction difference between $\overline{\theta}_{IU}$ and its adjacent beam $\eta_\ell$ $\left(0\leq \delta_u \leq \frac{1}{L}\right)$. Then, by denoting the best beam as $\boldsymbol{\phi}^\star=\mathbf{a}_r (\eta_\ell)$, the IRS beamforming gain can be expressed as
	\begin{equation} \label{exp_bs_irs}
		\left|\mathbf{a}_r^H (\overline{\theta}_{IU}) \boldsymbol{\phi}^\star \right| = \left| \sum\limits_{m=1}^M e^{j\pi \delta_u \left(-\frac{M-1}{2}+m-1\right)} \right|  = \frac{\sin (\frac{\pi M \delta_u }{2})}{\sin (\frac{\pi  \delta_u }{2})}.
	\end{equation}
	The function $\frac{\sin (\frac{\pi M x }{2})}{\sin (\frac{\pi  x }{2})}$ exhibits behavior similar to that of the $sinc$ function and has a zero value at $\frac{2}{M}$. This function decreases monotonically over $x \in [0,\frac{2}{M}]$.   Thus, when the user's angle $\overline{\theta}_{IU}$ is exactly aligned with the angle of the best beam, the IRS beamforming gain reaches its maximum value, i.e., $\delta_u=0$ and $\left|\mathbf{a}_r^H (\overline{\theta}_{IU}) \boldsymbol{\phi}^\star \right|=M$. When the user's angle $\overline{\theta}_{IU}$ lies in the middle of two adjacent beam angles, the IRS beamforming gain is the lowest, i.e., $\delta_u =\frac{1}{L}$ and  $\left|\mathbf{a}_r^T (\overline{\theta}_{IU}) \boldsymbol{\phi}^\star \right|={\sin \left(\frac{\pi M  }{2 L}\right)} {\sin^{-1} (\frac{\pi  }{2 L})}$. 
	
	After beam scanning, the user finds the IRS's best beam index $\ell$ and corresponding received signal power, and feeds them back to the IRS controller. In Phase II, the REs can then adopt  this beam to reflect the signals from the BS to the communication user  during the data transmission phase (if IRS single-beam sensing is used). The achievable rate of the user in bits per second (bps) by taking into account the beam scanning overhead is thus given by
	\begin{equation} \label{exp_sweeping}
		R  = \frac{T-\tau}{T} \log_2 \left(1+ \frac{N P_t |\alpha_g|^2 |\alpha_h|^2}{\sigma^2} \frac{\sin^2 (\frac{\pi M \delta_u  }{2  })}{\sin^2 (\frac{\pi \delta_u  }{2  })} \right).
	\end{equation}

	\subsection{CRB and MSE for Initial Sensing in Phase I} \label{sec_mse_phaseI}
	In this subsection, the target angle is first estimated via the maximum likelihood estimator (MLE). Then,  the CRB and an approximated closed-form expression for MSE associated with the angle estimation is derived. 
	The received echo signals at the   SEs in \eqref{exp_sensors2} by letting $s[t]=1, \forall t$ can be represented as
		\begin{align} \label{exp_sensor_phaseI}
		\hspace{-0.1cm}	\mathbf{y}_s[t] = & \sqrt{N P_t } \alpha_g \left( \mathbf{H}_t \operatorname{diag} (\boldsymbol{\phi}[t]) \mathbf{a}_r (\theta_{BI}) +\mathbf{a}_s(\theta_{BI}) \right)    + \mathbf{n}_s[t] \notag \\
			 = & \sqrt{N P_t} \alpha_g \left( \alpha_s \mathbf{a}_s(\theta_{IT}) \mathbf{a}_r^H (\theta_{IT}) \operatorname{diag} (\boldsymbol{\phi}[t]) \mathbf{a}_r(\theta_{BI}) \right. \notag\\ 
			 & +\left. \mathbf{a}_s(\theta_{BI}) \right)    + \mathbf{n}_s[t] \notag\\
			 = &\sqrt{N P_t} \alpha_g \left( \alpha_s \mathbf{a}_s(\theta_{IT})  \mathbf{a}_r^H (\overline{\theta}_{IT}) \boldsymbol{\phi}[t] \right. \notag \\
			 & + \left. \mathbf{a}_s(\theta_{BI}) \right)   +\mathbf{n}_s[t], 
		\end{align}
	where $\overline{\theta}_{IT}= \theta_{IT} - \theta_{BI}$. Note that the first part of $\mathbf{y}_s[t]$ is due to the BS-REs-target-SEs link, while the second part is due to the BS-SEs link. In order to estimate the target angle, the BS-SEs link signal should be first canceled out. Fortunately, the angle $\theta_{BI}$ can be estimated in advance since the positions of the BS and IRS are fixed. By exploiting the asymptotic orthogonality of the array steering vectors \cite{6831723}, the BS-REs-target-SEs link can be extracted from $\mathbf{y}_s[t]$. Specifically, the echo signals $\widehat{\mathbf{y}}_s[t]$ from the target can be extracted as
	\begin{equation}
		\widehat{\mathbf{y}}_s[t] = \left(\mathbf{I}_{M_s}- \frac{\mathbf{a}_s(\theta_{BI}) \mathbf{a}_s^H(\theta_{BI})}{M_s}\right)\mathbf{y}_s[t] .
	\end{equation}
	Note that by ignoring the noise, we have  
		\begin{align}
			\overline{\mathbf{y}}_s[t] \triangleq &
			\frac{\mathbf{a}_s(\theta_{BI}) \mathbf{a}_s^H(\theta_{BI})}{M_s}\mathbf{y}_s[t] \notag \\
			= & \sqrt{N P_t} \alpha_g \Big( \frac{\alpha_s}{M_s} \mathbf{a}_s(\theta_{BI}) \mathbf{a}_s^H(\theta_{BI}) \mathbf{a}_s(\theta_{IT})  \mathbf{a}_r^H (\overline{\theta}_{IT}) \boldsymbol{\phi}[t] \notag  \\
			& + \mathbf{a}_s(\theta_{BI}) \Big) ,  
		\end{align}
	due to $\frac{1}{M_s} \mathbf{a}_s^H(\theta_{BI}) \mathbf{a}_s(\theta_{BI}) =1$.
	Note that 
	\begin{equation} \label{func_sin}
		f(\overline{\theta}_{IT})\triangleq \left|\frac{\mathbf{a}_s^H(\theta_{BI}) \mathbf{a}_s(\theta_{IT})}{M_s}\right|^2= \frac{\sin^2 (\frac{\pi M_s \overline{\theta}_{IT} }{2})}{M_s^2 \sin^2 (\frac{\pi  \overline{\theta}_{IT} }{2})},
	\end{equation}
	whose mainlobe is located within $[0,\frac{2}{M_s}]$. When $\overline{\theta}_{IT}=0$, we have $f(\overline{\theta}_{IT})=1$. When $\overline{\theta}_{IT} = \frac{2k}{M_s}, k=1,2,\cdots, M_s-1$,  we have $f(\overline{\theta}_{IT})=0$. When $\frac{2}{M_s} \leq |\overline{\theta}_{IT}| \leq 2- \frac{2}{M_s}$,  $f(\overline{\theta}_{IT}) \leq 6.25\%$ in the case of $M_s\geq 5$. Thus, when $|\overline{\theta}_{IT}| \geq \frac{2}{M_s}$, we have $\overline{\mathbf{y}}_s[t]\approx \sqrt{N P_t} \alpha_g \mathbf{a}_s(\theta_{BI})  $. Therefore, when $0 \leq |\overline{\theta}_{IT}| < \frac{2}{M_s}$, which corresponds to the undetectable region $\Omega_u$ of the target, the echo signal from the target cannot be extracted from the BS-SEs direct signal, and thus the target angle cannot be estimated. Otherwise, the extracted signal can be further expressed as
	\begin{equation} \label{ys_extracted}
		\widehat{\mathbf{y}}_s[t] = \sqrt{N P_t} \alpha_g \alpha_s \mathbf{a}_s(\theta_{IT})   \mathbf{a}_r^H (\overline{\theta}_{IT}) \boldsymbol{\phi} [t]  + \mathbf{n}_s[t].
	\end{equation}
	
	By collecting all echo signals from the target during the $L$ symbol durations for  beam scanning, we have  
		\begin{align} \label{exp_collect_sweeping}
		\hspace{-0.1cm}	\mathbf{Y} &= [\widehat{\mathbf{y}}_s[1], \widehat{\mathbf{y}}_s[2],\cdots, \widehat{\mathbf{y}}_s[L] ] \notag \\
			& = \sqrt{N P_t} \alpha_g \alpha_s \mathbf{a}_s(\theta_{IT})   \mathbf{a}_r^H (\overline{\theta}_{IT}) \left[ \boldsymbol{\phi}[1], \boldsymbol{\phi}[2], \cdots, \boldsymbol{\phi}[L] \right]   +\mathbf{N} \notag \\
			& \triangleq \sqrt{N P_t} \alpha_g \alpha_s \mathbf{a}_s (\theta_{IT}) \mathbf{q}^H(\theta_{IT}) \mathbf{X} + \mathbf{N},
		\end{align}
	where  $\mathbf{X} \triangleq   \left[ \boldsymbol{\phi}[1], \boldsymbol{\phi}[2], \cdots, \boldsymbol{\phi}[L] \right]$, $\mathbf{q}(\theta_{IT}) \triangleq \mathbf{a}_r (\overline{\theta}_{IT})$, and $\mathbf{N} \triangleq [\mathbf{n}_s[1], \mathbf{n}_s[2], \cdots, \mathbf{n}_s[L] ]$. Let $\mathbf{R}_x  \triangleq \frac{1}{L} \mathbf{X} \mathbf{X}^H$ represent the covariance matrix of $\mathbf{X}$. With the codebook designed as in \eqref{exp_codebook}, we have $\mathbf{X}=\mathbf{D}$ and $\mathbf{R}_x =   \mathbf{I}_M$. 
	For ease of notation, we simply re-denote $\theta_{IT}$ by $\theta$.   Let $\boldsymbol{\xi} =[\theta, \text{Re}\{\alpha_s\}, \text{Im}\{\alpha_s\}]^T\in \mathbb{R}^{3\times 1}$ denote the vector of the unknown parameters to be estimated, which includes the target's angle and the complex channel coefficients. Particularly, we are interested in characterizing the MSE for estimating the angle. By vectorizing \eqref{exp_collect_sweeping}, we have
	\begin{equation} 
		\operatorname{vec}(\mathbf{Y}) = \alpha_s \operatorname{vec}(\mathbf{U}(\theta))+ \operatorname{vec}(\mathbf{N}),
	\end{equation}
	where $\mathbf{U}(\theta) =   \sqrt{N P_t} \alpha_g \mathbf{a}_s(\theta) \mathbf{q}^H(\theta)  \mathbf{X} $. Then, the target angle can be estimated according to the following theorem. 
	\begin{theorem} \label{lemma1}
		The angle estimated via the MLE is given by
		\begin{equation} \label{est_mle}
			\theta_\text{MLE} = \arg \max \limits_{\theta} \quad \left|\mathbf{a}_s^H (\theta) \mathbf{Y} \mathbf{X}^H \mathbf{q}(\theta) \right|^2,
		\end{equation}	
		which can be solved by exhaustive search over $\left[-1, 1 \right]$.
	\end{theorem}
	\begin{proof}
		See Appendix \ref{append_mle}.
	\end{proof} 
	
	\begin{figure}[t]
		\begin{centering}
			\includegraphics[width=.42\textwidth]{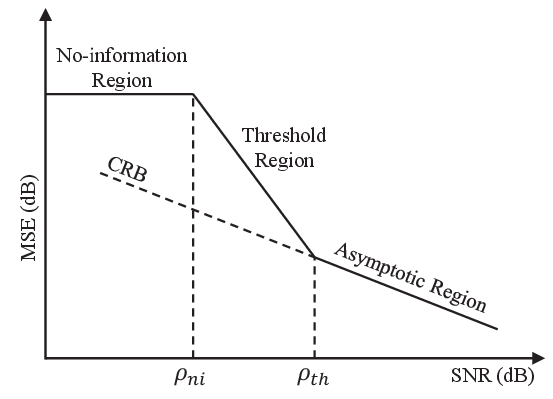}
			\caption{MSE versus SNR.}\label{fig_mse}
		\end{centering}
	\vspace{-0.5cm}
	\end{figure}
	Next, we derive the CRB and MSE of the angle estimation. The MSE curve, characterizing the performance of MLE, can typically be divided into three regions: asymptotic, threshold, and no-information regions \cite{1624646}, as depicted in Fig. \ref{fig_mse}. In the high SNR regime, the MSE is identical to the CRB, which is known as the   asymptotic region. As SNR decreases to a threshold $\rho_{th}$, the MSE starts to deviate from the CRB due to the presence of outliers, giving rise to the threshold region. Upon further decrease in SNR to a threshold $\rho_{ni}$, the desired signals become indistinguishable from the noise, resulting in a uniformly distributed estimation result across the entire parameter space, referred to as the no-information region. Therefore, a useful method to predict the MSE is the method of interval errors (MIE), which divides MSE into two parts: a local error term close to the true value (i.e., the CRB), and an outlier term accounting for global errors. The MIE approach is widely employed for MSE prediction under MLE\cite{1168170, 1344467, 1408188, 8681437}. Thus, we employ MIE to predict the MSE of angle estimation in this paper. Define $\hat{\theta}$ as the estimation of $\theta$, the MSE can be expressed as
		\begin{align}
			\mathrm{MSE} = &	\mathbb{E}[(\hat{\theta} - \theta)^2] \notag \\
			=& \Pr(\text{no outlier}) \mathbb{E}[(\hat{\theta}- \theta)^2 \big| \text{no outlier}] \notag \\
			&+ \Pr(\text{outlier}) \mathbb{E}[(\hat{\theta}- \theta)^2 \big| \text{outlier}],
		\end{align}
	where ``outlier" denotes the event that $\hat{\theta}$ falls outside  the mainlobe of the objective function. When MSE is located in the no-information region, we have $\hat{\theta} \sim U (-1,1)$ and $\mathbb{E}[(\hat{\theta}- \theta)^2 \big| \text{outlier}] = \frac{1}{3}$. When MSE is located in the asymptotic region, we have $\mathbb{E}[(\hat{\theta}- \theta)^2 \big| \text{no outlier}] = \mathrm{CRB}_\text{I}$ where $\mathrm{CRB}_\text{I}$ denotes the CRB of angle estimation in Phase I. Let $p=\Pr(\text{no outlier})$, the MSE can be rewritten as
	\begin{equation}\label{exp_mse_ori}
		\mathrm{MSE} = \frac{1-p}{3}+ p\cdot \mathrm{CRB}_\text{I}.
	\end{equation}
	In the sequel, we first derive $\mathrm{CRB}_\text{I}$, and then find the probability of the event of ``no outlier". 
	Let $\mathbf{F}\in \mathbb{R}^{3\times 3}$ denote the Fisher information matrix (FIM) for estimating $\boldsymbol{\xi}$. Since $\mathbf{N} \sim \mathcal{CN}(\mathbf{0}, \sigma^2 \mathbf{I}_{L M_s})$, each entry of $\mathbf{F}$ is given by \cite{kay1993fundamentals}   
	\begin{equation} 
		\vspace{-0.2cm}
		\mathbf{F}_{i,j} = \frac{2}{\sigma^2} \operatorname{Re}\left\{\frac{\partial \overline{\mathbf{U}}^H}{\partial \xi_i}  \frac{\partial \overline{\mathbf{U}}}{\partial \xi_j}\right\}, i,j \in\{1,2,3\},
	\end{equation}
where $\overline{\mathbf{U}} = \alpha_s \operatorname{vec}(\mathbf{U}(\theta))$.
	The FIM can be partitioned as
	\begin{equation}
		\vspace{-0.2cm}
		{\mathbf{F}}=	\begin{bmatrix}
		{\mathbf{F}}_{\theta \theta} & {\mathbf{F}}_{\theta \overline{\boldsymbol{\alpha}}} \\
		{\mathbf{F}}_{\theta \overline{\boldsymbol{\alpha}}}^{\mathrm{T}} & {\mathbf{F}}_{\overline{\boldsymbol{\alpha}} \overline{\boldsymbol{\alpha}}}
	\end{bmatrix},
	\end{equation}
	where $\overline{\boldsymbol{\alpha}}= [\operatorname{Re}\{\alpha_s\},\operatorname{Im}\{\alpha_s\}]^T$. The CRB for estimating the angle $\theta$ is defined as
	\begin{equation}
		\text{CRB}(\theta) = [\mathbf{F}^{-1}]_{1,1}=[\mathbf{F}_{\theta\theta} - \mathbf{F}_{\theta \overline{\boldsymbol{\alpha}}} \mathbf{F}_{\overline{\boldsymbol{\alpha}} \overline{\boldsymbol{\alpha}}}^{-1} \mathbf{F}_{\theta \overline{\boldsymbol{\alpha}}}^T]^{-1}.
	\end{equation}
	Then, the CRB for target sensing is given by the following theorem.
	\begin{theorem}
		The CRB for estimating the angle $\theta$ is given by
		\begin{equation} \label{exp_crb_ori}
			\begin{aligned}
				\mathrm{CRB}(\theta)
				&=\frac{\sigma^2}{2  |\alpha_s|^2\left(\operatorname{tr}\left(\dot{\mathbf{U}} (\theta)  \dot{\mathbf{U}}^{\mathrm{H}}(\theta) \right) - \frac{\left|\operatorname{tr}\left(\mathbf{U}(\theta) \dot{\mathbf{U}}^{\mathrm{H}}(\theta) \right)\right|^2} {\operatorname{tr}\left(\mathbf{U} (\theta) \mathbf{U}^{\mathrm{H}}(\theta) \right)}\right)}.	
			\end{aligned}
		\end{equation}
	\end{theorem}
	\begin{proof}
		See Appendix \ref{append_crb}.
	\end{proof}
	
	With the array response vector  defined as in  \eqref{exp_array_resp}, we obtain
	\begin{equation} \label{orth_as}
		\mathbf{a}_s^H (\theta) \dot{\mathbf{a}}_s(\theta) = 0,\quad \dot{\mathbf{a}}_s^H(\theta) \mathbf{a}_s(\theta) = 0,
	\end{equation}
	\begin{equation} \label{orth_q}
		\mathbf{q}^H (\theta) \dot{\mathbf{q}}(\theta) = 0,\quad \dot{\mathbf{q}}^H(\theta) \mathbf{q}(\theta) = 0, \forall \theta.
	\end{equation}
	Consequently, the CRB for estimating the angle $\theta$  can be simplified as
		\begin{align} \label{exp_crb_phaseI}
			\mathrm{CRB}_\text{I}&=\frac{\sigma^2}{2 L N P_t |\alpha_s|^2 |\alpha_g|^2 \left( M \|\dot{\mathbf{a}}_s(\theta)\|^2 + M_s \|\dot{\mathbf{q}}(\theta)\|^2  \right) } \notag \\
			&= \frac{6 \sigma^2}{L N P_t |\alpha_g|^2 |\alpha_s|^2 \pi^2 M M_s (M^2 + M_s^2 -2)} \notag \\
			&  \triangleq \frac{1}{\rho_t  \pi^2 N L} \frac{6}{ M M_s (M^2+ M_s^2-2)},
		\end{align}
	where $\rho_t\triangleq \frac{P_t |\alpha_g|^2 |\alpha_s|^2}{\sigma^2}$ represents the SNR of the target. For a given set of SNR and the number of BS antennas, we can improve the sensing accuracy by increasing the codebook size $L$, the number of   REs $M$, and the number of   SEs $M_s$, according to \eqref{exp_crb_phaseI}.	It is worth noting that the estimation of the spatial direction $\theta_{IT}$ exhibits an interesting characteristic: the resulting CRB remains unaffected by the spatial direction itself, as the employed codebook uniformly divides the entire spatial direction space. However, when estimating the physical angle $\zeta_{IT}$ of the target, the corresponding CRB becomes dependent on the specific physical angle. In addition, it can be observed from \eqref{exp_crb_phaseI} that  the number of REs $M$ and the number of SEs $M_s$ play equal roles in the CRB. Given that SEs have higher hardware cost, it would be more favorable to deploy more REs than SEs ($M > M_s$) if the total number of elements are fixed (i.e, $M+M_s$ is fixed). 
	
	Next, we find the probability of the event of ``no outlier", which can be determined according to the following theorem. 
	\begin{theorem}
		The probability of the event of ``no outlier" under the MLE can be approximated by
		\begin{equation} \label{p_no_outlier}
			p \approx \left(1-\frac{1}{2}\exp\left(-\frac{L \rho_t N M M_s}{2}\right)\right)^{M_s+M-2}.
		\end{equation}	
	\end{theorem}
	\begin{proof}
		See Appendix \ref{append_p_no_outlier}.
	\end{proof} 
	By substituting \eqref{exp_crb_phaseI} and \eqref{p_no_outlier} back into \eqref{exp_mse_ori}, we can obtain an approximate closed-form expression of MSE. With the obtained MSE, we have the following corollary.
	\begin{corollary} \label{corollary_1}
		The no-information threshold $\rho_{ni}$ of the MSE under the MLE can be approximated by
		\begin{equation}
			\rho_{ni} \approx -\frac{2}{L N M M_s} \ln \left[ 2\left( 1- (1-\beta_{ni})^{\frac{1}{M_s+M-2}}\right)\right],
		\end{equation}
		where $\beta_{ni}\in (0,1)$ is an empirical parameter, and is usually set as $\frac{8}{9}$ \cite{1168170}. The breakdown threshold $\rho_{th}$ of the MSE under the MLE can be approximated by
		\begin{equation}
			\rho_{th} \approx \frac{2}{L N M M_s} \left[ \ln \left(\rho_0 \right)+ \ln \left(\ln \left(\rho_0 \right) \right) \right],
		\end{equation}	
		where $\rho_0 = \frac{\pi^2 M^2 M_s^2 (M_s+M-2)}{2}$.
	\end{corollary}
	\begin{proof}
		The proof is similar to that in \cite{1168170} and thus omitted here due to space limitation.
	\end{proof} 

	Based on Corollary 1, it is observed that when the  SNR of the target surpasses $\rho_{th}$, the MSE coincides with the CRB. If the  SNR lies within the range of $\rho_{ni}$ and $\rho_{th}$, the MSE deviates from the CRB. Conversely, when the  SNR falls below $\rho_{ni}$, the MSE becomes independent of the  SNR. Consequently, in order to achieve accurate estimation of the target angle, it is imperative for the  SNR to exceed $\rho_{th}$.

	\section{Enhanced Sensing during Data Transmission} \label{sec_target_enhancement}
	 {At the end of IRS beam scanning process, the communication user identifies the IRS's best beam index $\ell$ and the corresponding received signal power, and feeds them back to the IRS controller where the corresponding user's SNR  $\gamma_\ell$ can be obtained.} The SEs perform an initial estimation of the target angle according to Theorem \ref{lemma1}. Subsequently, two cases arise in the data transmission phase. If the estimated angle is close to the communication user, all REs should reflect the data signal towards the user with the IRS's best beam, which also helps achieve high SNR at the target to enhance the estimation 	accuracy (i.e., IRS single-beam). On the other hand, if the angles of the user and the target with respect to the IRS are well separated, beam splitting could be conducted such that some of REs adopt a different beam to provide beamforming gain to improve the sensing accuracy, provided that the achievable rate of the user with the remaining REs is still satisfactory. In this section, we investigate the enhancement of target angle estimation when the user's achievable rate has sufficient margin, considering the above two IRS beamforming strategies, respectively.
	
	\subsection{Single Beam Based Sensing} \label{sec_no_split}
	A straightforward strategy for REs is to adopt the best beam  $\boldsymbol{\phi}^\star=\mathbf{a}_r (\eta_\ell)$ for the communication user. In this case, the echo signals from the target in \eqref{ys_extracted} can be rewritten as
	\begin{equation} 
		\widehat{\mathbf{y}}_s[t] = \sqrt{N P_t} \alpha_g \alpha_s \mathbf{a}_s(\theta_{IT})   \mathbf{a}_r^H (\overline{\theta}_{IT}) \boldsymbol{\phi}^\star s[t]+ \mathbf{n}_s[t].
	\end{equation}	
	The IRS beamforming gain for the target can be expressed as
	\begin{equation} 
		\left|\mathbf{a}_r^H (\overline{\theta}_{IT}) \boldsymbol{\phi}^\star \right| = \left| \sum\limits_{m=1}^M e^{j\pi \delta_t \left(-\frac{M-1}{2}+m-1\right)} \right|  = \frac{\sin (\frac{\pi M \delta_t }{2})}{\sin (\frac{\pi  \delta_t }{2})},
	\end{equation}	
	where $\delta_t \triangleq |\overline{\theta}_{IT}- \eta_\ell|$ denotes the spatial direction difference between $\overline{\theta}_{IT}$ and the best beam $\eta_\ell$. As described in Section \ref{sec_mse_phaseI}, this IRS beamforming gain is  small when $\delta_t \geq \frac{2}{M}$. Therefore, the region where the target estimation can benefit from this strategy is given by $\{ \Omega_t| \theta_{IT}: |\theta_{IT} - \theta_{BI} - \eta_\ell| < \frac{2}{M}, \theta_{IT} \notin \Omega_u\}$. In the following, we first estimate the target angle via the MLE, and then derive its CRB.
	
	By collecting the echo signals from the target in Phase II, we have 
		\begin{align} \label{exp_echoes_co2}
			\mathbf{Y}_2 =& \left[\widehat{\mathbf{y}}_s[\tau+1], \widehat{\mathbf{y}}_s[\tau+2],\cdots, \widehat{\mathbf{y}}_s[\tau+ \tau_2]\right] \notag \\
			=& \sqrt{N P_t} \alpha_g \alpha_s \mathbf{a}_s (\theta_{IT})   \mathbf{q}^H (\theta_{IT}) \boldsymbol{\phi}^\star \notag \\
			&\times [s[\tau+1],s[\tau+2], \cdots, s[\tau+ \tau_2]] + \mathbf{N}_2 \notag \\
			\triangleq & \sqrt{N P_t} \alpha_g \alpha_s  \mathbf{a}_s (\theta_{IT})   \mathbf{q}^H (\theta_{IT}) \mathbf{X}_2 + \mathbf{N}_2,
		\end{align}
	where $\mathbf{X}_2 \triangleq \boldsymbol{\phi}^\star [s[\tau+1],s[\tau+2], \cdots, s[\tau+\tau_2]] \in \mathbb{C}^{M\times \tau_2}$, $\tau_2 = T- \tau$, and $\mathbf{a}_r(\overline{\theta}_{IT})$ is replaced by $\mathbf{q}(\theta_{IT})$ as in \eqref{exp_collect_sweeping}.  The data signals are assumed to be independent of each other and satisfy $\mathbb{E}(|s[\tau_i]|^2) = 1$ for $\tau_i\in (\tau,\tau+\tau_2]$. Thereby, $\mathbf{X}_2 \mathbf{X}_2^H = \tau_2 \boldsymbol{\phi}^\star (\boldsymbol{\phi}^\star)^H$. Then, we collect the echo signals in Phase I and Phase II together to estimate the angle of the target.   By collecting  \eqref{exp_collect_sweeping} and  \eqref{exp_echoes_co2}, we have
	\begin{equation} \label{exp_echoes_all}
	\hspace{-0.1cm}	\overline{ \mathbf{Y}} = \sqrt{N P_t} \alpha_g \alpha_s \mathbf{a}_s (\theta_{IT})   \left[\mathbf{q}^H (\theta_{IT}) \mathbf{X}, \mathbf{q}^H(\theta_{IT}) \mathbf{X}_2  \right] +[\mathbf{ N}, \mathbf{ N}_2].
	\end{equation}
	Accordingly, the target angle can be estimated based on the following theorem.
	\begin{theorem}
		The angle estimated via the MLE by combining the echo signals in Phase I and Phase II together is given by
		\begin{equation}  \label{exp_mse_II}
			\theta_\text{MLE} = \arg \max \limits_{\theta} \quad \frac{\left|\mathbf{a}_s^H (\theta) \mathbf{Y} \mathbf{X}^H \mathbf{q}(\theta) + \mathbf{a}_s^H(\theta) \mathbf{Y}_2 \mathbf{X}_2^H \mathbf{q}(\theta) \right|^2} {L M+ (T-\tau) \left| \mathbf{q}^H(\theta) \boldsymbol{\phi}^\star \right|^2},
		\end{equation}	
		which can be found by exhaustive search over $[-1,1]$.
	\end{theorem}
	\begin{proof}
		The proof is similar to that in Appendix \ref{append_mle}. Thus, it is omitted here due to space limitation.
	\end{proof} 

	Note that the MLE needs to know the transmitted signals $\mathbf{X}_2$ during Phase II, which is not available due to random data signals. However, there are alternative estimation methods available that do not rely on the knowledge of  transmitted signals in practice. One such method is the Multiple Signal Classification (MUSIC) algorithm, which can approach MLE performance in high SNR scenarios without the need for information about the transmitted signals \cite{17564}. 
	
	Next, we consider the CRB of angle estimation via the MLE. We can adopt the MIE employed in Section \ref{sec_mse_phaseI} to predict the MSE. However, due to the existence of the denominator term in \eqref{exp_mse_II}, the chi-square distribution cannot be adopted for the MLE. Consequently, the method utilized in Theorem 3 to calculate the probability of the ``no outlier" event cannot be applied in this scenario. Hence, we only explore the CRB (instead fo MSE) of angle estimation.
	For notational convenience, we omit $\theta_{IT}$ here. 	Let $\mathbf{U}_2 \triangleq \sqrt{N P_t} \alpha_g \mathbf{a}_s   \left[\mathbf{q}^H  \mathbf{X}, \mathbf{q}^H  \mathbf{X}_2  \right] $. We thus have
	\begin{equation} 
	\hspace{-0.1cm}	\dot{\mathbf{U}}_2 = \sqrt{N P_t} \alpha_g \dot{ \mathbf{a}_s} \left[\mathbf{q}^H \mathbf{X}, \mathbf{q}^H \mathbf{X}_2  \right]+ \sqrt{N P_t} \alpha_g \mathbf{a}_s \left[\dot{ \mathbf{q}}^H \mathbf{X}, \dot{\mathbf{q}}^H \mathbf{X}_2\right].
	\end{equation}
	According to Theorem 2,  {the CRB for angle estimation in the whole phase (i.e., by combining the echo signals in Phase I and Phase II together) can be obtained as follows},
	\begin{align} \label{crb_whole}
	\hspace{-0.1cm} \mathrm{CRB}_\text{w}  =& \frac{\sigma^2} {2|\alpha_s|^2 \left( \operatorname{tr}(\dot{\mathbf{U}}_2 \dot{\mathbf{U}}_2^H) - \frac{|\operatorname{tr}( {\mathbf{U}_2} \dot{\mathbf{U}}_2^H)|^2}{\operatorname{tr}( {\mathbf{U}}_2 {\mathbf{U}}_2^H)} \right) } \notag \\
	=& \frac{\sigma^2} {2|\alpha_s|^2} \bigg(\beta_1 M \|\dot{\mathbf{a}_s}\|_2^2 + \beta_2 |\mathbf{q}^H \boldsymbol{\phi}^\star|^2 \|\dot{\mathbf{a}_s}\|_2^2 + \beta_1 M_s \|\dot{\mathbf{q}}\|_2^2 \notag \\
	&+ \beta_2 M_s |\dot{\mathbf{q}}^H \boldsymbol{\phi}^\star|^2 - \frac{|\beta_2|^2 M_s^2 |\mathbf{q}^H \boldsymbol{\phi}^\star (\boldsymbol{\phi}^\star)^H \dot{\mathbf{q}} |^2}{\beta_1 M_s \|\mathbf{q}\|_2^2 + \beta_2 M_s  |\mathbf{q}^H \boldsymbol{\phi}^\star|^2} \bigg)^{-1} \notag  \\
	\overset{(a)}{\leq} & \frac{\sigma^2} {2|\alpha_s|^2} \frac{1}{\beta_1 M \|\dot{\mathbf{a}_s}\|_2^2 + \beta_1 M_s \|\dot{\mathbf{q}}\|_2^2+ \beta_2 |\mathbf{q}^H \boldsymbol{\phi}^\star|^2 \|\dot{\mathbf{a}_s}\|_2^2 } \notag \\
	  = &\frac{1}{\rho_t \pi^2 N } \notag \\
	  &\times  \frac{6}{ L MM_s (M^2+ M_s^2-2) + \tau_2 |\mathbf{q}^H \boldsymbol{\phi}^\star|^2 M_s(M_s^2-1)} \notag \\
	 \triangleq& \mathrm{CRB}_\text{up},
\end{align}	
	where $\mathrm{CRB}_\text{w}$ denotes the CRB for angle estimation in the whole phase,  $\beta_1 \triangleq  N P_t |\alpha_g|^2 L$, $\beta_2 \triangleq N P_t |\alpha_g|^2 \tau_2$, $(a)$ holds by ignoring the term $\beta_1 M_s \|\mathbf{q}\|_2^2$, the equality holds when $\mathbf{q} = \boldsymbol{\phi}^\star$, and $|\mathbf{q}^H \boldsymbol{\phi}^\star|^2 =  \left|\frac{\sin ({\pi M \delta_t }/{2})}{\sin ({\pi  \delta_t }/{2})} \right|^2$. 
	
	Compared with the CRB in \eqref{exp_crb_phaseI}, it is observed that  the first component of $\mathrm{CRB}_\text{up}$ (i.e., $L M M_s (M^2+ M_s^2-2)$) represents the contribution of Phase I estimation, and the second component of $\mathrm{CRB}_\text{up}$ (i.e., $\tau_2 |\mathbf{q}^H \boldsymbol{\phi}^\star|^2 M_s (M_s^2-1)$) represents the contribution of Phase  II estimation. The inequality $\mathrm{CRB}_\text{w} \leq \mathrm{CRB}_\text{up}< \mathrm{CRB}_\text{I} $ indicates that the estimation accuracy is improved by combining Phase I and Phase II together as compared to Phase I. Notably, there are distinct behaviors exhibited by the  REs in Phase II compared to Phase I. On  one hand, the  REs conduct beam scanning in Phase I, thus providing more flexibility to sense the target. Consequently, even if the number of  SEs is one (i.e., $M_s=1$), we can still estimate the target angle as in \eqref{exp_crb_phaseI}. On the other hand, the  REs fix their beamforming direction during   Phase II, instead of perusing further beam scanning. As such, if the  number of  SEs is one in Phase II, the second component of $\mathrm{CRB}_\text{up}$ becomes zero, and thus the sensing performance cannot be improved effectively.

	\subsection{Beam Splitting Based Sensing} \label{sec_beam_split}
	The single-beam-based sensing is effective only when the target locates in the vicinity of the communication user. To further enhance target estimation, we propose IRS beam splitting when the communication user's achievable rate has sufficient margin. Specifically, REs can be divided into two groups: one for communication and the other for target sensing, with different beams, respectively. However, beam splitting can potentially degrade the communication SNR and introduce interference between the two  beams. Therefore, it becomes important to determine the conditions and the portion of  REs for applying the beam splitting.

	We first examine the impact of beam splitting on data transmission. Suppose $M_e$  REs  are allocated for target angle estimation, while the remaining $M-M_e$ REs are dedicated to data transmission in Phase II. Since the initial target angle is estimated in Phase I, the $M_e$ REs are specifically designed to align with the estimated angle in order to maximize the signal power at that target. The remaining $M-M_e$ elements are set to be the last $M-M_e$ entries of the IRS's best beam to serve the communication user. For the ease of analysis, we assume that the user exactly locates at its best beam direction (i.e., $\delta_u = 0$), and the following analysis can be extended to the case where $\delta_u \neq 0$.  When all  REs are used for communication, the IRS's best beam is given by
	\begin{equation}
		\boldsymbol{\phi}^\star = e^{-j\frac{ \pi (M-1)  \overline{\theta}_{IU}} {2}} \left[\begin{array}{llll}
			1 & e^{j  \pi \overline{\theta}_{IU}}  & \ldots & e^{j  \pi (M-1) \overline{\theta}_{IU}}
		\end{array}\right]^{T}.
	\end{equation}
	Then, we have $|\mathbf{a}_r^H (\overline{\theta}_{IU}) \boldsymbol{\phi}^\star| = M$. When $M_e$ REs are allocated for target angle estimation, the reflection coefficients of  REs are given by
	\begin{align} \label{irs_partition}
		\boldsymbol{\phi}_e = \; & e^{-j\frac{ \pi (M-1)  \overline{\theta}_{IU}} {2}} \Big[ 
		1 , e^{j  \pi \overline{\theta}_{IT}} , \ldots , e^{j  \pi (M_e-1)  \overline{\theta}_{IT}}, \notag \\
		& e^{j  \pi M_e  \overline{\theta}_{IU}} , \ldots , e^{j  \pi (M-1)  \overline{\theta}_{IU}}
		\Big]^{T}.
	\end{align}	
	Then, the IRS beamforming gain for the communication user can be obtained as
	\begin{equation} \label{g_irs}
		G_\text{IRS} \triangleq |\mathbf{a}_r^H (\overline{\theta}_{IU}) \boldsymbol{\phi}_e| = \left| M-M_e + \varrho \frac{\sin \left(\frac{\pi M_e \delta_{UT}}{2}\right)} {\sin \left(\frac{\pi \delta_{UT}}{2}\right)} \right|,
	\end{equation}		
	where $\delta_{UT} = |\theta_{IU} - \theta_{IT}|$ denotes the spatial direction difference between the user and target, and $\varrho = \exp\left(j \pi \delta_{UT}\frac{(M_e-1)}{2} \right)$. Due to the existence of  $\varrho$ and $\sin(\pi M_e \delta_{UT}/2)$,   $G_\text{IRS}$ may be smaller than $M-M_e$, which means that the $M_e$ REs allocated for target estimation may negatively impact the communication rate performance. Thus, it is crucial to determine the condition under which  $G_\text{IRS}$ does not deviate much from $M-M_e$. Then, we have
		\begin{align}
			G_\text{IRS}^2 =& (M-M_e)^2 +  \frac{\sin^2 \left(\frac{\pi M_e \delta_{UT}}{2}\right)} {\sin^2 \left(\frac{\pi \delta_{UT}}{2}\right)} \notag \\
			&+ 2 (M-M_e) \frac{\sin \left(\frac{\pi M_e \delta_{UT}}{2}\right) \cos\left(\frac{\pi (M_e-1) \delta_{UT}}{2}\right)} {\sin \left(\frac{\pi \delta_{UT}}{2}\right)} \notag \\
			\overset{(a)}{\approx} & (M-M_e)^2 +  \frac{\sin^2 \left(\frac{\pi M_e \delta_{UT}}{2}\right)} {\sin^2 \left(\frac{\pi \delta_{UT}}{2}\right)} \notag \\
			&+ \frac{2(M-M_e)}{\pi \delta_{UT}} \sin(\pi M_e \delta_{UT}),
		\end{align} 
	where $(a)$ holds due to $2\sin \left(\frac{\pi M_e \delta_{UT}}{2}\right) \cos\left(\frac{\pi (M_e-1) \delta_{UT}}{2}\right) \approx \sin(\pi M_e \delta_{UT})$ for $M_e \gg 1$, and $\sin \left(\frac{\pi \delta_{UT}}{2}\right) \approx \frac{\pi \delta_{UT}}{2}$ for $\delta_{UT} \ll 1$. We find that the effect on $G_\text{IRS}$ mainly stems from $\sin(\pi M_e \delta_{UT})$. When $\sin(\pi M_e \delta_{UT})= -1$, i.e., $M_e \delta_{UT}= \frac{3}{2}+ 2k, k=0,1,2,\ldots$,   $G_\text{IRS}$ can be much smaller than $M-M_e$, and thus degrade the communication SNR dramatically. It is also observed that when $M_e \delta_{UT} > 2$, the value of $\frac{\sin \left({\pi M_e \delta_{UT}}/{2}\right)} {\sin \left({\pi \delta_{UT}}/{2}\right)}$ in \eqref{g_irs} will be very small. Consequently, the threshold of $\delta_{UT}$ is selected as $\frac{11}{2M_e}$. When $\delta_{UT} > \frac{11}{2M_e}$, the effect of beam splitting on the communication can be neglected. Note that the effect of beam splitting on the target is similar to \eqref{g_irs}. Therefore, the threshold of $\delta_{UT}$ is set to $\frac{11}{M}$. When $\delta_{UT} > \frac{11}{M}$,  it can be regarded that the $M_e$  REs and the other $M-M_e$  REs cause only negligible interference to each other. Overall, the condition for choosing beam splitting is when the target angle falls in the region $\{ \Omega_e| \theta_{IT}: |\theta_{IT} - \theta_{BI} - \eta_\ell| > \frac{11}{M}, \theta_{IT} \notin \Omega_u\}$ \footnote{ Note that the above analysis assumes that the user exactly locates at the optimal beam direction, i.e., $\delta_u=0$. When $\delta_u\neq 0$, we have $\delta_{UT} \in \{|\overline{\theta}_{IT} - \eta_\ell|- \delta_u, |\overline{\theta}_{IT} - \eta_\ell|+ \delta_u\}$. In the worse case, we have $\delta_{UT}= |\overline{\theta}_{IT} - \eta_\ell|- \frac{1}{L}> \frac{10}{M}$, where the user and target are also well separated. This is the reason why we do not choose $\delta_{UT}$  as $\frac{7}{2M_e}$, which is not sufficient in general to ensure the user and target to be well separated.}. Fig. \ref{g_irs1} illustrates $G_\text{IRS}$ versus $\delta_{UT}$ when  $M_e$ is set to   different values. It is found that $\frac{11}{M}$ is a good threshold. When $\delta_{UT}>\frac{11}{M}$,   $G_\text{IRS}$ fluctuates only slightly around $M-M_e$. 
	
\begin{figure}[t]
	\begin{centering}
		\includegraphics[width=.42\textwidth]{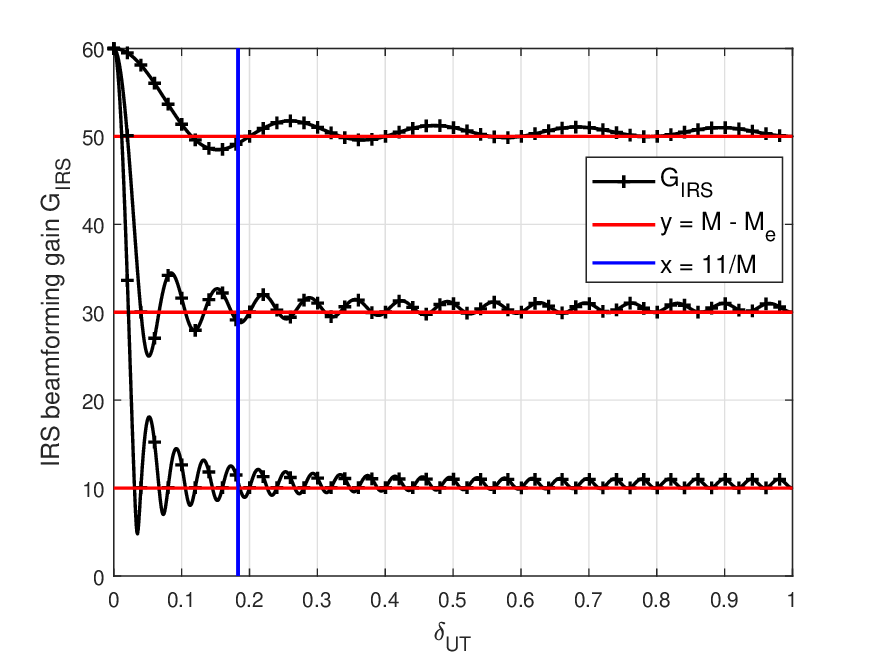}
		\caption{IRS beamforming gain $G_\text{IRS}$ versus $\delta_{UT}$ when $M=60$.}\label{g_irs1}
	\end{centering}
	\vspace{-0.3cm}
\end{figure}

	
	
	Next, we determine the element number for beam splitting. The user feeds the IRS's best beam index $\ell$ and corresponding received signal power back to the IRS controller where the corresponding user's SNR $\gamma_\ell$ can be obtained. If the SNR $\gamma_\ell$ of the communication user is larger than the predefined target $\gamma$, the beam splitting can be performed in the region $\Omega_e$. However, the exact angle of the user is unavailable with only the knowledge of the best beam $\eta_\ell$. Thus, we consider the worst case to decide the element number for beam splitting. Specifically, the user's SNR is given by
		\begin{equation}
		\gamma_\ell = \frac{N P_t |\alpha_g|^2 |\alpha_h|^2}{\sigma^2} \frac{\sin^2 (\frac{\pi M \delta_u  }{2  })}{\sin^2 (\frac{\pi \delta_u  }{2  })} \triangleq G_\text{ch} \frac{\sin^2 (\frac{\pi M \delta_u  }{2  })}{\sin^2 (\frac{\pi \delta_u  }{2  })},
	\end{equation}	
	where $G_\text{ch}\triangleq \frac{N P_t |\alpha_g|^2 |\alpha_h|^2}{\sigma^2}$ is the unknown channel gain. Since the expression ${\sin^2 (\frac{\pi M \delta_u  }{2  })}/{\sin^2 (\frac{\pi \delta_u  }{2  })}$ decreases as $\delta_u$ increases within the interval $[0,1/L]$, we first assume that the user exactly locates at the optimal beam direction (i.e., $\delta_u=0$) to obtain the worst channel gain $G_\text{ch,min} = \frac{\gamma_\ell}{M^2}$ under a given $\gamma_\ell$. Then, we set $\delta_u$ in \eqref{exp_gammae} to $1/L$ to obtain the worst IRS beamforming gain.  Therefore, when $M_e$ REs are split for target sensing, the worst channel gain multiplies the worst IRS beamforming gain, leading to the worst SNR $\gamma_e$ of the user as follows,
	\begin{equation} \label{exp_gammae}
		\gamma_e = G_\text{ch,min} \frac{\sin^2 (\frac{\pi (M-M_e) \delta_u }{2 })}{\sin^2 (\frac{\pi \delta_u }{2  })} = \frac{\gamma_\ell}{M^2} \frac{\sin^2 (\frac{\pi (M-M_e)  }{2 L })}{\sin^2 (\frac{\pi  }{2 L  })} \geq \gamma.
	\end{equation}	
	Therefore, the element number for target sensing is given by
	\begin{subequations}
		\begin{align} 
			M_e & \leq \left\lfloor M- \frac{2L}{\pi} \arcsin \left(M \sqrt{\frac{\gamma}{\gamma_\ell}} \sin \left(\frac{\pi}{2L}\right)\right) \right\rfloor \label{num_split} \\
			& \overset{(a)}{\approx} \left\lfloor M\left(1- \sqrt{\frac{\gamma}{\gamma_\ell}}\right) \right\rfloor ,
		\end{align}	
	\end{subequations}
	where $(a)$ holds due to $\sin(x)\approx x$ and $\arcsin(x)\approx x$ for $x \ll 1$. In order to improve the sensing accuracy, $M_e$ can be set as the right-hand side of \eqref{num_split}. It is noted that when the user actually locates at the optimal beam direction, we can obtain $M_e = \left\lfloor M\left(1- \sqrt{\frac{\gamma}{\gamma_\ell}}\right) \right\rfloor$. 
	
	Finally, we consider the performance of target estimation after beam splitting. Following the similar procedure in Section \ref{sec_no_split}, we can obtain the corresponding CRB of angle estimation. Consequently, the CRB is similar to that in \eqref{crb_whole}, with $\mathbf{q}^H \boldsymbol{\phi}^\star$ replaced by $\mathbf{q}^H \boldsymbol{\phi}_e$.
	
		\begin{figure*}[t] 
	\begin{minipage}[t]{0.45\linewidth} 
		\centering
		\includegraphics[width=\textwidth]{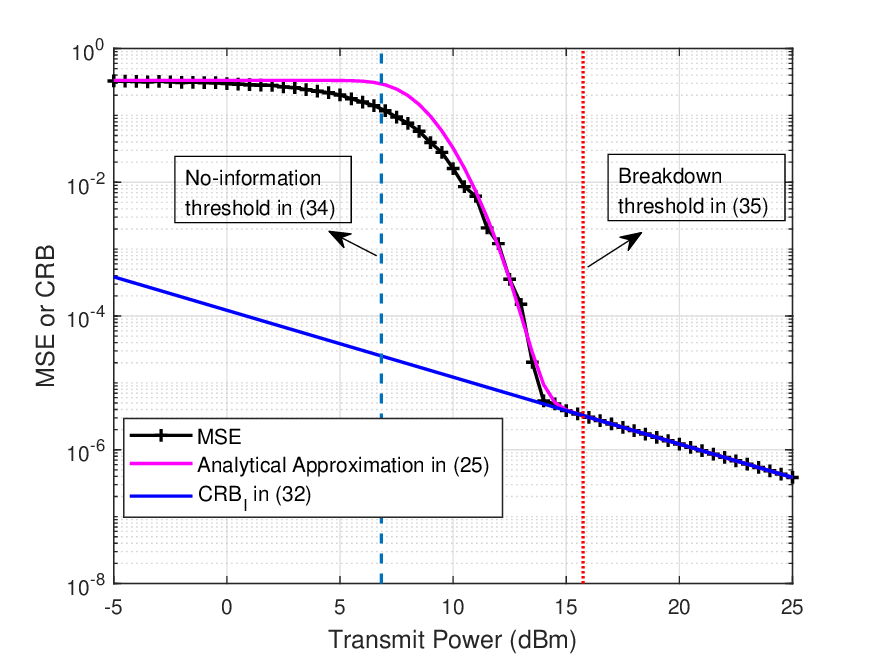}
		\caption{MSE versus transmit power in Phase I when $M=64$ and $M_s=12$.}\label{fig_mseI_1}
	\end{minipage}%
	\hspace{0.9cm}
	\begin{minipage}[t]{0.45\linewidth}
		\centering
		\includegraphics[width=\textwidth]{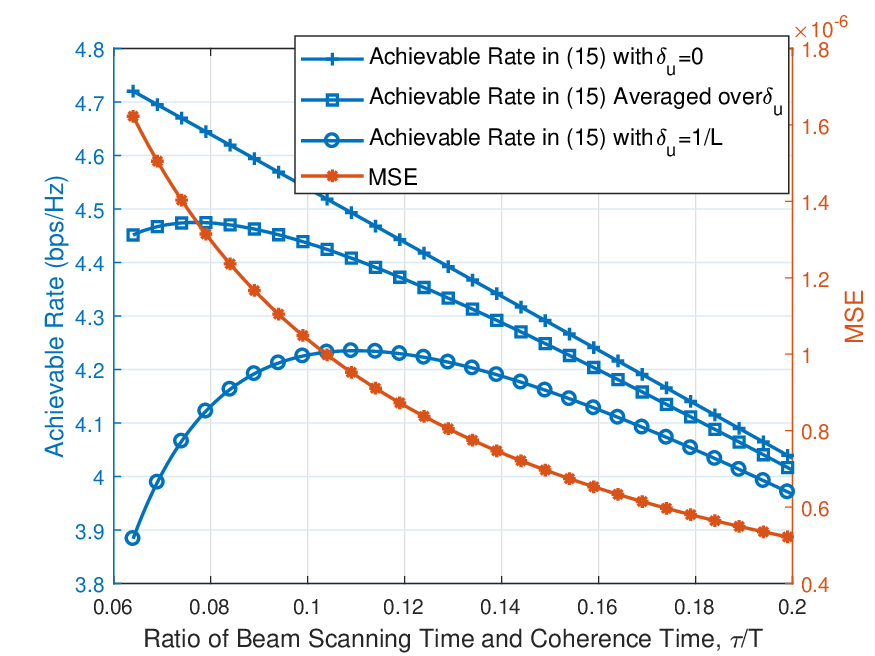}
		\caption{Achievable rate and MSE versus ratio of beam scanning time and coherence time with $P_t=20$ dBm.}\label{fig_online_m}
	\end{minipage}
	\vspace{-0.3cm}
\end{figure*}	
	\section{Simulation Results}
	In this section, numerical examples are provided to validate our analysis and evaluate the performance of our proposed protocol. The carrier frequency is $f_c=28$ GHz. Other system parameters are set as follows unless specified otherwise later: $N=64$, $M=64$, $M_s=12$, $L=M$, $P_t=30$ dBm, $T=1000$, $d_{BI}=30$ m, $d_{IU}=10$ m, $d_{IT}=5$ m, $\zeta_{BI} = -60^\circ$, $\zeta_{IU} = 0^\circ$, $\zeta_{IT}=30^\circ$, $\sigma^2 = -120$ dBm, and $\kappa = 7$ dBsm. The curves of MSE are obtained by averaging over 1000 independent realizations of the noise.

	\subsection{Communication and Sensing Trade-off in Phase I}	
	\vspace{-0.2cm}
	We first consider the MSE performance of MLE in Phase I. 	Fig. \ref{fig_mseI_1} illustrates a comparison of the analytical approximation of MSE in \eqref{exp_mse_ori} and the actual MSE with respect  to the transmit power when $M=64$ and $M_s=12$. There are some interesting findings. Firstly, the derived analytical approximation closely matches the actual MSE in both low-SNR and high-SNR regimes. When the transmit power falls in the threshold region, the derived analytical approximation slightly overestimates the actual MSE. Overall, the derived analytical approximation performs well.  Secondly, Corollary \ref{corollary_1} can well predict the no-information threshold and breakdown threshold. Thus, in order to ensure reliable angle estimation in Phase I, the transmit power should exceed the breakdown threshold. 
	
	{Next, we investigate the communication and sensing trade-off in Phase I.  Fig. \ref{fig_online_m} shows the achievable rate and MSE versus the ratio of beam scanning time and coherence time, when all  REs are adopted for communication during data transmission period. The curve labeled ``Achievable Rate in (15) Averaged  over $\delta_u$" is generated by taking the expectation of \eqref{exp_sweeping} over $\delta_u$, where $\delta_u \sim \mathcal{U}(0,\frac{1}{L})$. Since exhausting beam scanning is adopted, we have $\tau=L$ as in Section \ref{sec_rate_phaseI}. The achievable rate and MSE exhibit different variations versus the IRS beam scanning time. As the time of beam scanning  increases, the IRS beamforming gain \eqref{exp_bs_irs} in the case of $\delta_u=\frac{1}{L}$ increases at the expense of reduced data transmission time, leading to an initial increase and subsequent decrease in the achievable rate. In the case of $\delta_u=0$, the IRS beamforming gain already reaches its peak when $L=M$. Thus, increasing the beam scanning time further only leads to a reduction in the achievable rate. However, when considering the average effect, the achievable rate initially rises and subsequently declines. On the other hand, the MSE monotonically decreases  as the time of beam scanning increases. Thus, a proper beam scanning time that balances achievable rate and MSE is desired. However, with our proposed enhanced sensing approaches in Section \ref{sec_target_enhancement}, we can first ensure the communication quality in Phase I, and then further improve sensing accuracy in Phase II.}
	
	\subsection{Communication and Sensing Trade-off in Phase II}
	We first consider the MSE performance in Phase II as depicted in Fig. \ref{fig_mseII_1}, where the target angle ($\zeta_{IT}=0^\circ$) is close to the communication user's angle. Several interesting findings are inferred from the results. Firstly, the concise form of $\mathrm{CRB}_{\text{up}}$ closely approximates $\mathrm{CRB}_{\text{w}}$, making it convenient for analysis. Secondly, when the transmit power is below 5 dBm, the curve of $\mathrm{CRB}_{\text{up}}$ is nearly the same as that of $\mathrm{CRB}_{\text{I}}$. This occurs because the MSE falls within the no-information region, resulting in significant estimation errors in Phase I and $|\mathbf{q}^H \boldsymbol{\phi}^\star|^2 \approx 0$. As a result, the sensing accuracy cannot be improved in Phase II. Thirdly, the no-information threshold and breakdown threshold in Phase II closely resemble those in Phase I since the  REs allocated for sensing in Phase II point their beam towards the angle estimated in Phase I. Consequently, if the angle estimation in Phase I is inaccurate, the  REs for sensing cannot point their beam towards the target effectively in Phase II, thus resulting in limited performance improvement.  
	
	\begin{figure*}[t] 
		\begin{minipage}[t]{0.45\linewidth} 
			\centering
			\includegraphics[width=\textwidth]{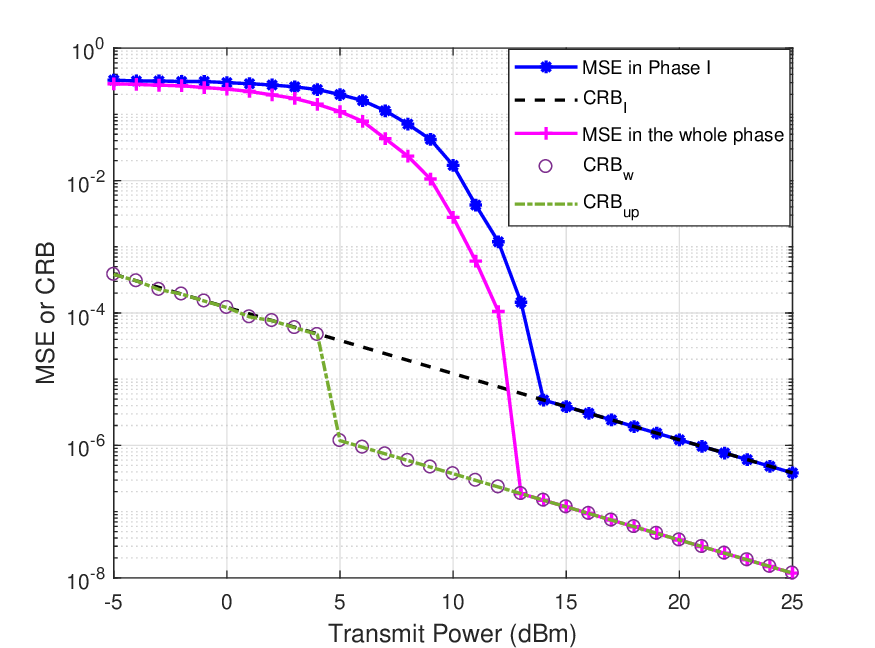}
			\caption{MSE versus transmit power in Phase II when $\zeta_{IU}=0^\circ$ and $\zeta_{IT}=0^\circ$.}\label{fig_mseII_1}
		\end{minipage}%
		\hspace{0.9cm}
		\begin{minipage}[t]{0.45\linewidth}
			\centering
			\includegraphics[width=\textwidth]{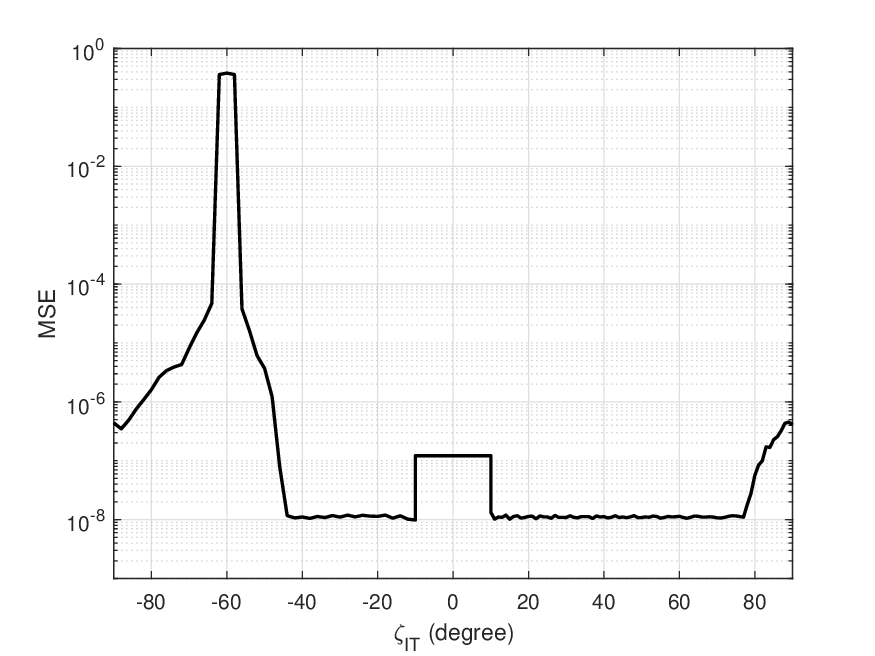}
			\caption{MSE versus the target angle in beam-splitting-based sensing when $M=64$ and $M_e=36$.}\label{re_strategy2}
		\end{minipage}
	\vspace{-0.3cm}
	\end{figure*}
		\begin{figure*}[t] 
		\begin{minipage}[t]{0.45\linewidth} 
			\centering
			\includegraphics[width=\textwidth]{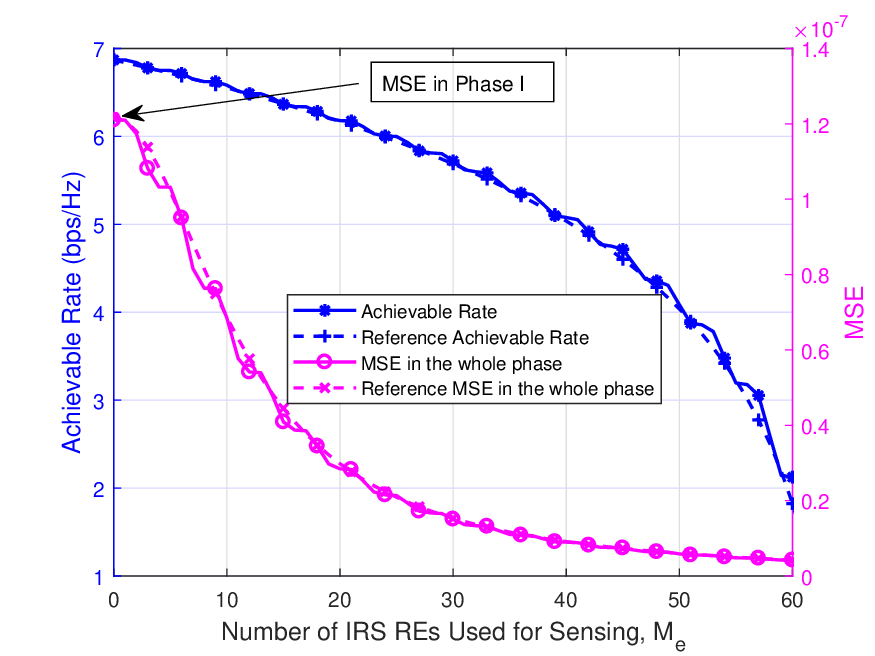}
			\caption{Achievable rate and MSE versus the number of  REs used for sensing  when $P_t = 30$ dBm, $\zeta_{IU}=0^\circ$ and $\zeta_{IT}=30^\circ$.}\label{rate_crb_Nr_re1}
		\end{minipage}%
		\hspace{0.9cm}
		\begin{minipage}[t]{0.45\linewidth}
			\centering
			\includegraphics[width=\textwidth]{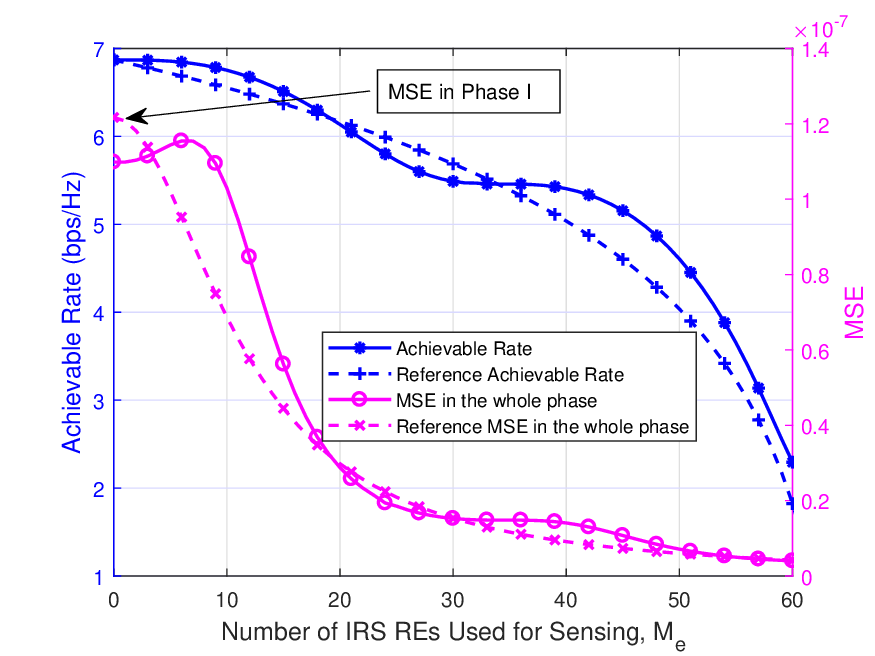}
			\caption{Achievable rate and MSE versus the number of  REs used for sensing  when $P_t = 30$ dBm, $\zeta_{IU}=0^\circ$ and $\zeta_{IT}=3^\circ$.}\label{rate_crb_Nr_re2}
		\end{minipage}
	\vspace{-0.3cm}
	\end{figure*}
	
	Next, the influence of the undetectable region on the sensing performance is illustrated in Fig. \ref{re_strategy2}, which depicts  the MSE versus the target angle in beam-splitting-based sensing when $M=64$ and $M_e=36$. It is observed that the MSE is high when $\theta_{IT}<-44^\circ$ or $\theta_{IT}> 75^\circ$, primarily due to the presence of the undetectable region. 
	Specifically, the undetectable region is defined as $\{\Omega_u | -90^\circ< \zeta_{IT}< -44.4^\circ \; \text{or} \; 75.3^\circ< \zeta_{IT}< 90^\circ\}$, where $\frac{180}{\pi} \arcsin\left(\sin\left(\frac{-60\pi}{180}\right)+ \frac{2}{12}\right) \approx -44.4^\circ$ and $\frac{180}{\pi} \arcsin\left(\sin\left(\frac{-60\pi}{180}\right)- \frac{2}{12}+2\right) \approx 75.3^\circ$. In addition, the MSE increases for $|\zeta_{IT}|<10^\circ$ due to the beam splitting only conducted at $\arcsin(11/M) = 9.90^\circ$. To better present the MSE performance, we focus on the region $\{\zeta_{IT}|-40^\circ \leq \zeta_{IT} \leq 70^\circ\}$ in the subsequent investigation.

	We compare the achievable rate and MSE versus the number of  REs used for sensing in Fig. \ref{rate_crb_Nr_re1} and Fig. \ref{rate_crb_Nr_re2} where  $\zeta_{IT}$ is set to $30^\circ$ and $3^\circ$, respectively. The curve labeled ``Reference Achievable Rate" represents the achievable rate obtained when $M-M_e$ REs are dedicated to communication without any interference. The curve labeled ``Reference MSE in the whole phase" represents the MSE obtained when $M_e$ REs are dedicated to target sensing without any interference in Phase II.  {The leftmost point on the curve of ``Reference MSE in the whole phase" represents the MSE in Phase I.}   It is firstly observed that  the achievable rate shows a decreasing trend, while the estimation accuracy  exhibits an increasing trend as the number of  REs allocated for sensing increases. When the user and target are well separated (Fig. \ref{rate_crb_Nr_re1}), the achievable rate and MSE remain relatively close to their references, respectively. It is worth noting that the slight fluctuation in the curve labeled ``Achievable Rate" is due to the minimal impact of the REs split for sensing on the communication performance. This observation is consistent with the results shown in Fig. \ref{g_irs1}, where slight fluctuation occurs when $\delta_{UT}>\frac{11}{M}$. 
	However, when the user and target are close (Fig. \ref{rate_crb_Nr_re2}), beam splitting results in a significant drop in the achievable rate compared to its reference around $M_e=25$. The MSE in the whole phase fluctuates more severely around the reference MSE, but it is always smaller than the MSE in Phase I. Therefore, beam-splitting-based sensing can only be adopted when the user and target are well separated to ensure the communication quality, as discussed in Section \ref{sec_beam_split}. 
	
	\begin{figure*}[t] 
		\begin{minipage}[t]{0.45\linewidth} 
			\centering
			\includegraphics[width=\textwidth]{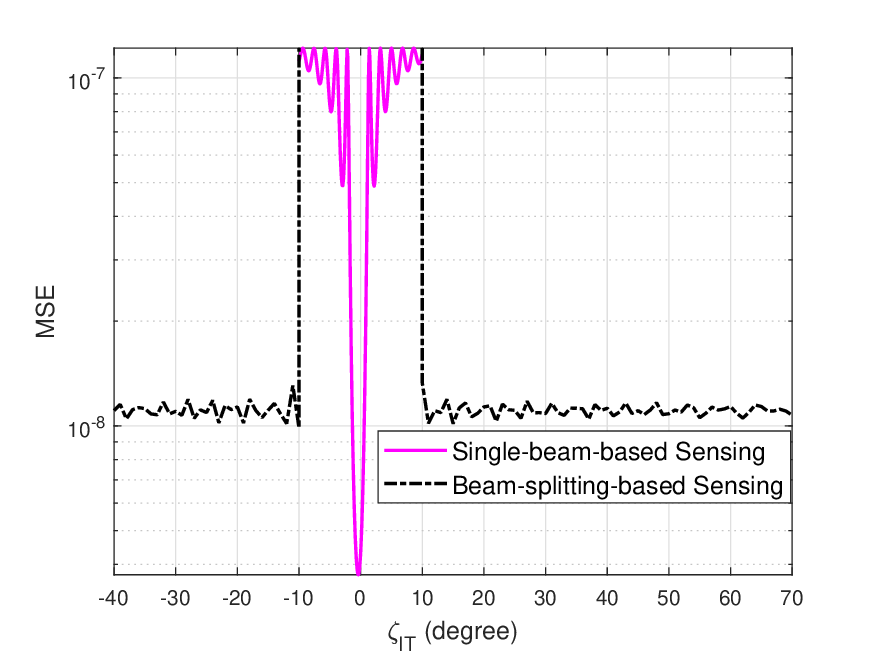}
			\caption{MSE versus the target angle in single-beam-based sensing and beam-splitting-based sensing when $P_t = 30$ dBm.}\label{re_strategy12_mse}
		\end{minipage}%
		\hspace{0.9cm}
		\begin{minipage}[t]{0.45\linewidth}
			\centering
			\includegraphics[width=\textwidth]{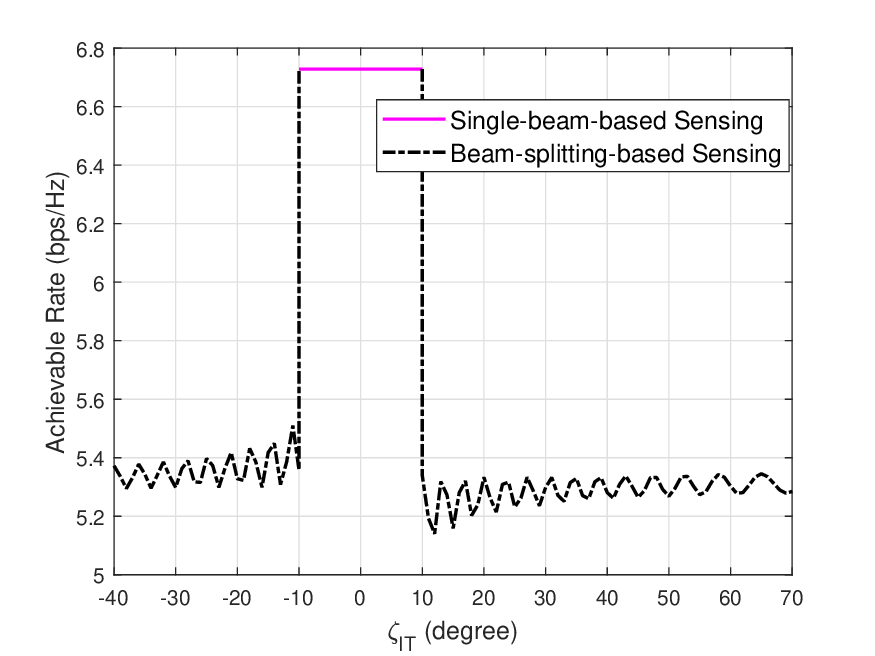}
			\caption{Achievable rate versus the target angle in single-beam-based sensing and beam-splitting-based sensing when $P_t = 30$ dBm.}\label{re_strategy12_rate}
		\end{minipage}
	\vspace{-0.3cm}
	\end{figure*}
	
	Next, we proceed to evaluate the MSE versus the target angle in the single-beam-based sensing and beam-splitting-based sensing as displayed in Fig. \ref{re_strategy12_mse} when the transmit power is $P_t=30$ dBm. The corresponding achievable rate is shown in Fig. \ref{re_strategy12_rate}, with a threshold achievable rate of the communication user set at 5.0 bps/Hz. The number of  REs split for target sensing is determined as $M_e = 36$ according to \eqref{num_split}. Several interesting findings are observed.  Firstly, in Fig. \ref{re_strategy12_mse}, the MSE decreases significantly around $\zeta_{IT}=0^\circ$ in the single-beam-based sensing. This behavior can be attributed to the region $\{\Omega_t \big||\zeta_{IT}|< \arcsin(2/M)\approx 1.80^\circ\}$. Secondly, the MSE in the single-beam-based sensing also decreases for $|\zeta_{IT}|<10^\circ$ due to the sidelobes of the function $\frac{\sin ({\pi M x }/{2})}{\sin ({\pi x }/{2})}$. However, the MSE exhibits little improvement since the target is outside the region $\Omega_t$, and the signals cannot effectively reach the target. Thirdly, in the beam-splitting-based sensing, the MSE decreases for $|\zeta_{IT}|>10^\circ$ due to the beam splitting conducted at $\arcsin(11/M) \approx 9.90^\circ$. However, since only a portion of  REs is allocated for target sensing during Phase II, the sensing performance in the beam-splitting-based sensing is inferior to that achieved at the region $\Omega_t$ in the single-beam-based sensing.
	
	As for the achievable rates shown in Fig. \ref{re_strategy12_rate}, the single-beam-based sensing maintains a constant achievable rate since all  REs reflect the signals towards the user's direction in Phase II. Conversely, the achievable rate in the beam-splitting-based sensing decreases from 6.7 bps/Hz to around 5.3 bps/Hz when $|\zeta_{IT}|>10^\circ$. Nevertheless, this rate remains above the 5.0 bps/Hz threshold.  Consequently, the single-beam-based sensing benefits the target sensing when the angles of the user and target are close to each other. On the other hand, the beam-splitting-based sensing improves the sensing accuracy at the cost of reducing the achievable rate of the communication user, which is advantageous when the user and   target are well separated in angle and the achievable rate of the user has sufficient margin.
	
	\begin{figure*}[t] 
		\begin{minipage}[t]{0.44\linewidth} 
			\centering
			\includegraphics[width=\textwidth]{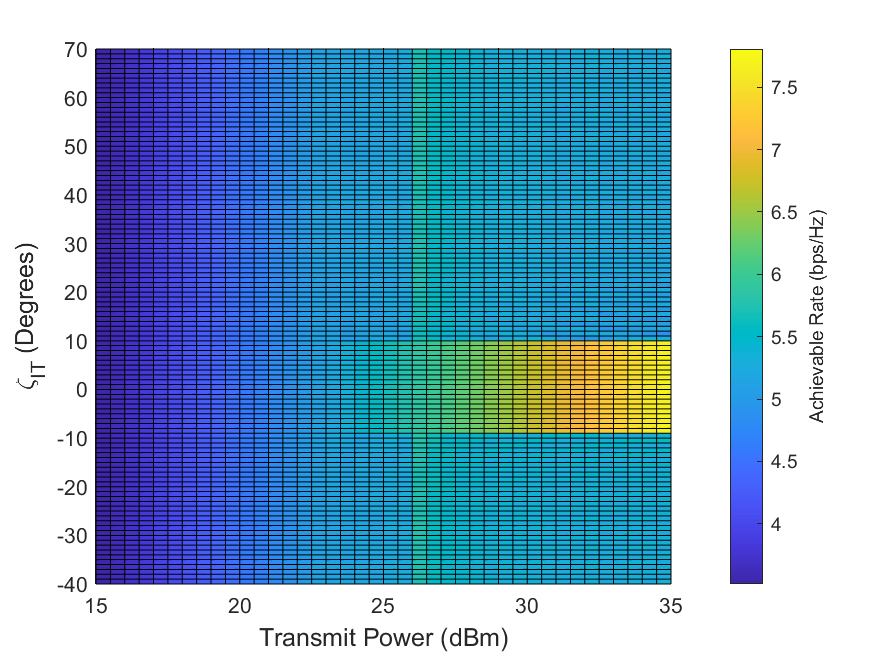}
			\caption{Achievable rate versus the target angle and transmit power.}\label{strategy3_pt_angle2}
		\end{minipage}%
		\hspace{1cm}
		\begin{minipage}[t]{0.44\linewidth}
			\centering
			\includegraphics[width=\textwidth]{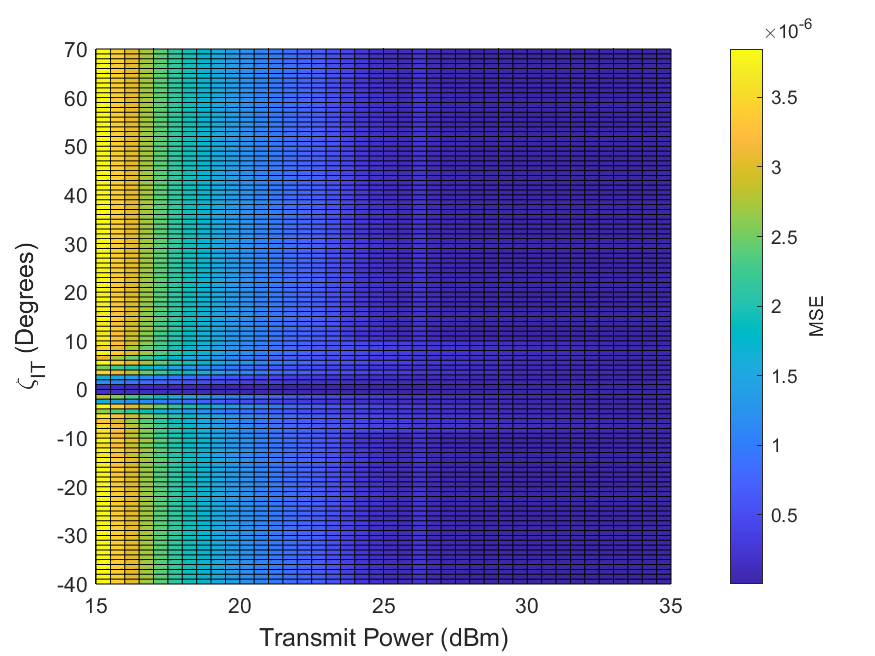}
			\caption{MSE versus the target angle and transmit power.}\label{strategy3_pt_angle2_mse}
		\end{minipage}
	\vspace{-0.6cm}
	\end{figure*}
	
	Furthermore, we explore the achievable rate and MSE  with respect to transmit power and target angle in Fig. \ref{strategy3_pt_angle2} and Fig. \ref{strategy3_pt_angle2_mse}, respectively. It is noteworthy that when the transmit power is below 22 dBm, the achievable rate remains below 5 bps/Hz, indicating insufficient rate margin for target sensing during Phase II by beam splitting based sensing. Hence, the sensing performance can only be improved in the region $\Omega_t$ with single-beam based sensing. As the transmit power increases, beam splitting can be employed, resulting in further enhanced sensing performance in the region $\Omega_e$.

	\section{Conclusion}
	In this paper, we have proposed a new ISAC protocol for an IRS-aided mmWave system  that utilizes downlink  beam scanning/data signals for achieving simultaneous beam training and target sensing. We derive the achievable rate of the communication user 	and the CRB/MSE of the target angle estimation in the beam scanning and data transmission phases, respectively. In particular, two IRS beam design and sensing strategies, namely, single-beam-based sensing and beam-splitting-based sensing, are proposed to enhance the sensing accuracy during the data transmission phase while ensuring the communication quality. Numerical results have verified the effectiveness of the proposed protocol and design.
	
	\appendix
	\begin{appendices}
		\subsection{Proof of Theorem 1} \label{append_mle}
		Denoting $\operatorname{vec}(\mathbf{Y})$ as $\widetilde{\mathbf{y}}$, the likelihood function of $\operatorname{vec}(\mathbf{Y})$ given $\boldsymbol{\xi}$ is
		\begin{small}
			\begin{equation} 
				L(\widetilde{\mathbf{y}}; \boldsymbol{\xi}) = \frac{1}{\left(\pi \sigma^2\right)^{L M_s}} \exp \left(-\frac{1}{\sigma^2} \|\widetilde{\mathbf{y}} - \alpha_s \operatorname{vec}(\mathbf{U}(\theta)) \|^2\right).
			\end{equation}
		\end{small}
	
		Then, maximizing the likelihood function is equivalent to minimizing $\|\widetilde{\mathbf{y}} - \alpha_s \operatorname{vec}(\mathbf{U}(\theta)) \|^2$. Therefore, the MLE can be written as
		\begin{equation} \label{exp_mle}
			\begin{aligned}
				(\theta_{\text{MLE}}, \alpha_{\text{MLE}}) = \arg \min \limits_{\theta, \alpha} \|\widetilde{\mathbf{y}} - \alpha_s \operatorname{vec}(\mathbf{U}(\theta)) \|^2.
			\end{aligned}
		\end{equation}
		With any given $\theta$, the optimal $\alpha$ is given by $\alpha_{\text{MLE}} = \frac{(\operatorname{vec}(\mathbf{U}(\theta)))^H \widetilde{\mathbf{y}}}{\|\operatorname{vec}(\mathbf{U}(\theta))\|^2}$.
		By substituting $\alpha_{\text{MLE}}$ back into \eqref{exp_mle}, yielding
		\begin{small}
			\begin{align}\label{exp_mle_theta}
				\hspace{-0.1cm}\|\widetilde{\mathbf{y}} - \alpha_{\text{MLE}} \operatorname{vec}(\mathbf{U}(\theta)) \|^2 
				& = \|\widetilde{\mathbf{y}}\|^2 - \frac{\left|(\operatorname{vec}(\mathbf{U}(\theta)))^H \operatorname{vec}(\mathbf{Y})\right|^2} {\|\operatorname{vec}(\mathbf{U}(\theta))\|^2} \notag \\
				& =  \|\widetilde{\mathbf{y}}\|^2 - \frac{\left|\mathbf{a}_s^H (\theta) \mathbf{Y} \mathbf{X}^H \mathbf{q}(\theta) \right|^2} {L M M_s}.
			\end{align}
		\end{small}
		Thereby, the MLE of $\theta$ is given by \eqref{est_mle} .
		
		\subsection{Proof of Theorem 2} \label{append_crb}
Since $\frac{\partial \overline{\mathbf{U}}}{\partial \theta}  = \alpha_s \operatorname{vec} (\dot{\mathbf{u}}(\theta) )$ and $\frac{\partial \overline{\mathbf{U}}}{\partial \overline{\boldsymbol{\alpha}}} = [1,j]\otimes \operatorname{vec} ({\mathbf{u}(\theta)})$, we have
\begin{small}
	\begin{equation}
		\mathbf{F}_{\theta\theta} = \frac{2}{\sigma^2} \operatorname{Re}\left\{\alpha_s \left(\operatorname{vec} (\dot{\mathbf{U}} ) \right)^H \alpha_s \operatorname{vec} (\dot{\mathbf{U}}  ) \right\}
		= \frac{2 |\alpha_s|^2}{\sigma^2} \operatorname{tr}  \left(\dot{\mathbf{U}}  \dot{\mathbf{U}}^H \right).
	\end{equation}
\end{small}
\begin{small}
	\vspace{-0.5cm}
\begin{align}
	\mathbf{F}_{\theta \overline{\boldsymbol{\alpha}}} &= \frac{2}{\sigma^2} \operatorname{Re}\left\{\alpha_s^*  \operatorname{vec} (\dot{\mathbf{U}}^H)     [1,j]\otimes \operatorname{vec} ({\mathbf{U}}  ) \right\}	\notag		 \\
	&= \frac{2}{\sigma^2} \operatorname{Re}\left\{ \alpha_s^* [1,j]\otimes \left(\operatorname{vec} (\dot{\mathbf{U}}^H  ) \operatorname{vec} ({\mathbf{U}}  ) \right) \right\} \notag \\
	& = \frac{2}{\sigma^2} \operatorname{Re}\left\{ \alpha_s^*  \operatorname{tr} \left(\mathbf{U} \dot{\mathbf{U}}^H \right)    [1,j] \right\}.
\end{align}
\end{small}
\begin{small}
				\begin{align}
		\mathbf{F}_{\overline{\boldsymbol{\alpha}} \overline{\boldsymbol{\alpha}}} 
		&= \frac{2}{\sigma^2} \operatorname{Re}\left\{\left([1,j]\otimes \operatorname{vec} ({\mathbf{U}} ) \right)^H [1,j]\otimes \operatorname{vec} ({\mathbf{U}}  ) \right\}		\notag	 \\
		&= \frac{2}{\sigma^2} \operatorname{Re}\left\{ \left( [1,j]^H[1,j] \right) \otimes \left(\operatorname{vec}(\mathbf{U}^H) \operatorname{vec}(\mathbf{U})\right) \right\} \notag \\
		& = \frac{2}{\sigma^2} \operatorname{Re}\left\{ \left( [1,j]^H[1,j] \right) \otimes \left(\operatorname{tr} (\mathbf{U}^H \mathbf{U})\right) \right\} \notag \\
		& = \frac{2}{\sigma^2} \operatorname{tr}\left(\mathbf{U}   \mathbf{U}^H\right) \mathbf{I}_2.
	\end{align}
\end{small}

		Thus, the FIM can be obtained as in  \eqref{exp_crb_ori}.	
		
		\subsection{Proof of Theorem 3} \label{append_p_no_outlier}
		Based on the MLE in Theorem 1, we have 
\begin{small}
	\begin{align}
		I(\theta-\theta_0) \triangleq & \left|\mathbf{a}_s^H (\theta) \mathbf{Y} \mathbf{X}^H \mathbf{q}(\theta) \right|^2 \notag \\
		= & \Big|\mathbf{a}_s^H (\theta) \big(\sqrt{N P_t} \alpha_g \alpha_s \mathbf{a}_s(\theta_0) \mathbf{q}(\theta_0)^H \mathbf{X} + \mathbf{ N}\big) \mathbf{X}^H \mathbf{q}(\theta) \Big|^2 \notag \\
		=& \Big|  \sqrt{N P_t} L \alpha_g \alpha_s \mathbf{a}_s^H (\theta) \mathbf{a}_s(\theta_0) \mathbf{q}(\theta_0)^H \mathbf{q}(\theta)  \notag \\
		&+  \mathbf{a}_s^H (\theta) \mathbf{ N}  \mathbf{X}^H \mathbf{q}(\theta) \Big|^2 \notag \\
		\triangleq & \left|  \sqrt{N P_t} L \alpha_g \alpha_s f_{M_s}(\theta-\theta_0) f_{M}(\theta-\theta_0)  + w \right|^2, 
	\end{align}
\end{small}where $\theta_0$ is the actual angle to be estimated, $f_{M}(\theta) \triangleq \frac{\sin(\pi M\theta/2)}{\sin(\pi \theta/2)}$, and $w \triangleq \mathbf{a}_s^H (\theta) \mathbf{ N}  \mathbf{X}^H \mathbf{q}(\theta) \sim \mathcal{CN}(0,L M M_s \sigma^2)$. Therefore, $\frac{2}{L M M_s \sigma^2} I(\theta-\theta_0)$ is a non-central chi-square distributed random variable with two degrees of freedom with non-centrality parameter given by
\begin{small}
				\begin{align}
		\Upsilon &= \frac{2N P_t L^2 |\alpha_g|^2 |\alpha_s|^2}{L M M_s \sigma^2} f_{M_s}^2(\theta-\theta_0) f_{M}^2(\theta-\theta_0) \notag \\
		&= 2L \rho_t N  M M_s \frac{f_{M_s}^2(\theta-\theta_0)}{M_s^2} \frac{f_{M}^2(\theta-\theta_0)}{M^2}.
	\end{align}
\end{small}

		Therefore, the periodogram sampled at discrete point $\{\frac{2k}{M_s} \big| k=0,1,\ldots, M_s-1\}$ and $\{\frac{2i}{M} \big| i=0,1,\ldots, M-1\}$ is distributed according to
\begin{small}
	\begin{equation}
	\hspace{-0.2cm}	\frac{2 I\left(\theta-\theta_0\right)}{L M M_s \sigma^2} \sim\left\{\begin{array}{cl}
			\chi_2^2(2L \rho_t N M M_s), & \theta-\theta_0=0, \\
			\chi_2^2, & \theta-\theta_0=\frac{2}{M_s} k, k\neq 0, \\
			\chi_2^2, & \theta-\theta_0=\frac{2}{M} i, i\neq 0,
		\end{array}\right.
	\end{equation}
\end{small}where $\chi_2^2$ and $\chi_2^2(\cdot)$ represent the central and non-central  chi-square distributions with two degrees of freedoms, respectively. Defining $\widetilde{I}(k) \triangleq \frac{2 I\left(\frac{2k}{M_s}\right)}{L M M_s \sigma^2}$, we have $\widetilde{I}(k)$ is $\chi_2^2$ for $k\neq 0$ and $\chi_2^2(2L \rho_t N M M_s)$ for $k=0$. Then, the cdf of  $\widetilde{I}(k)$ is given by
\begin{small}
			\begin{equation}
		F_{\widetilde{I}(k)} (x) = \left\{\begin{array}{cc}
			1-Q_1(\sqrt{2L \rho_t N M M_s }, \sqrt{x}), & k =0 , \\
			1-\exp(-x/2), & k\neq 0,
		\end{array}\right.
	\end{equation}
\end{small}where $Q_1(\alpha,\beta)$ denotes the first-order Marcum-Q function with parameter $\alpha$ and $\beta$. Similarly, defining $\overline{I}(i) \triangleq \frac{2 I\left(\frac{2i}{M}\right)}{L M M_s \sigma^2}$, we have $\overline{I}(i)$ is $\chi_2^2$ for $i\neq 0$ and $\chi_2^2(2L \rho_t N M M_s)$ for $i=0$. Then, the cdf of  $\overline{I}(i)$ is given by
\begin{small}
			\begin{equation}
		F_{\overline{I}(i)} (y) = \left\{\begin{array}{cc}
			1-Q_1(\sqrt{2L \rho_t N M M_s }, \sqrt{y}), & i =0 , \\
			1-\exp(-y/2), & i\neq 0.
		\end{array}\right.
	\end{equation}
\end{small}

		Thereby, the probability of the event of ``no outlier" can be expressed as
\begin{small}
			\begin{equation} \label{exp_pro_apen}
		\begin{aligned}
			\hspace{-0.1cm}	p &=\Pr\left\{\widetilde{I}(0)>\max\left(\widetilde{I}(k)\right)\right\} \Pr\left\{\overline{I}(0)>\max\left(\overline{I}(i)\right)\right\}  \\
			&=\prod\limits_{k=1}^{M_s-1} \Pr\left\{\widetilde I(0)>\widetilde I(k)\right\} \prod\limits_{i=1}^{M-1} \Pr\left\{\overline I(0)>\overline I(i)\right\},k\neq0, i\neq 0.
		\end{aligned}
	\end{equation}
\end{small}

		We consider the first part of the above probability. Since $\widetilde{I}(k), k\neq 0$, are i.i.d. random variables,  the first part of the probability can be simplified to
\begin{small}
				\begin{align}
		p_1 &= \left(\Pr\left\{\widetilde I(0)>\widetilde I(k)\right\} \right)^{M_s-1} \notag \\
		&=\left[\int_0^\infty p_{\widetilde{I}(k)}(x)\left(1-F_{\widetilde{I}(0)}(x)\right)dx\right]^{M_s-1} , \; k\neq 0,
	\end{align}
\end{small}where $p_{\widetilde{I}(k)}(x)$ is the pdf of the exponentially distribution. Then, we have
\begin{small}
				\begin{align}
		p_1 &=\left(\int_{0}^{\infty}\frac{1}{2}e^{-x/2}Q_{1}\left(\sqrt{2 L \rho_t N M M_s},\sqrt{x}\right)d x\right)^{M_s-1} \notag \\
		&=\left(\int_0^\infty x e^{-x^2/2}Q_1\left(\sqrt{2 L \rho_t N M M_s},x\right)dx\right)^{M_s-1}  \notag \\
		&= \left(1-\frac{1}{2}\exp\left(-\frac{L \rho_t N M M_s}{2}\right)\right)^{M_s-1}.
	\end{align}
\end{small}Since the second part of  \eqref{exp_pro_apen} is similar to the first part, the probability of the event of ``no outlier" can be approximated by
\begin{small}
			\begin{equation}
		p \approx \left(1-\frac{1}{2}\exp\left(-\frac{L \rho_t N M M_s}{2}\right)\right)^{M_s+M-2}.
	\end{equation}	
\end{small}

	\end{appendices}
	
	\bibliographystyle{IEEEtran}
	\bibliography{ris_split}
	
\end{document}